\documentclass[11pt]{article}
\usepackage{mathrsfs}
\usepackage{amsmath}
\usepackage[colorlinks=true, allcolors=blue, urlcolor=black]{hyperref}

\makeatletter
\newcommand{\removelatexerror}{\let\@latex@error\@gobble}
\makeatother
\usepackage{amsmath}
\usepackage{amsthm}
\usepackage{amsfonts}
\usepackage{amssymb}
\usepackage{graphicx}
\usepackage{dsfont}
\usepackage{mathtools}
\usepackage{cleveref}
\usepackage{pgfplots}
\usepackage{bbm}
\usepackage[noend]{algpseudocode}
\usepackage{stmaryrd}

\usepackage{algorithm}
\usepackage{complexity}
\usepackage{enumitem}
\usepackage{tikz}
\usetikzlibrary{positioning}
\usepackage{thm-restate}
\usepackage{float}
\usetikzlibrary{calc}
\usepackage{thm-restate}
\usetikzlibrary{shapes,decorations.markings}
\usepackage{caption}
\usepackage{subcaption}
\usepackage{comment}

\usepackage{lipsum}
\usetikzlibrary{snakes}

\newtheorem{theorem}{Theorem}[section]
\newtheorem{observation}[theorem]{Observation}
\newtheorem{assumption}[theorem]{Assumption}
\newtheorem{definition}[theorem]{Definition}
\newtheorem{claim}[theorem]{Claim}
\newtheorem{lemma}[theorem]{Lemma}
\newtheorem{proposition}[theorem]{Proposition}

\newtheorem{invariant}[theorem]{Invariant}

\newtheorem{fact}[theorem]{Fact}
\newtheorem{corollary}[theorem]{Corollary}

\newenvironment{proofsketch}[1]{\noindent{\it Proof Sketch.} {#1}{\qed\bigskip}}

\renewcommand{\E}{\mathbb{E}}
\newcommand{\tO}{\widetilde{O}}

\renewcommand{\W}{\mathcal{W}}
\usepackage[english]{babel}
\usepackage[T1]{fontenc}
\usepackage{caption}

\newcommand{\maxflow}[0]{\textnormal{MaxFlow}}
\newcommand{\mincut}[0]{\textnormal{MinCut}}
\newcommand{\eps}[0]{\varepsilon}

\newcommand{\FindEdge}{\mathsf{FindEdge}}
\newcommand{\IS}{\mathsf{IS}}

\usepackage[margin=1in]{geometry}

\algblockdefx{Withprob}{EndWithprob}[1]{\textbf{with probability} #1 \textbf{do}}{}

\makeatletter
\ifthenelse{\equal{\ALG@noend}{t}}%
  {\algtext*{EndWithprob}}
  {}%
\makeatother

\title{Minimum $s$--$t$ Cuts with Fewer Cut Queries}

\author{Yonggang Jiang\footnote{MPI-INF \ \href{}{\url{yjiang@mpi-inf.mpg.de}}} \and Danupon Nanongkai\footnote{MPI-INF \ \href{}{\url{danupon@gmail.com}}} \and Pachara Sawettamalya\footnote{Department of Computer Science, Princeton University. Supported by NSF CAREER award CCF-233994. \ \href{mailto:ps3122@princeton.edu}{\url{ps3122@princeton.edu}}}}

\date{}

\usepackage[backend=biber, style=alphabetic, doi=false, url=false, isbn=false, minalphanames = 3, maxalphanames = 4, maxbibnames=10]{biblatex}
\begin{document}

\maketitle


\begin{abstract}

We study the problem of computing a minimum $s$--$t$ cut in an unweighted, undirected graph via \emph{cut queries}. In this model, the input graph is accessed through an oracle that, given a subset of vertices $S \subseteq V$, returns the size of the cut $(S, V \setminus S)$.

This line of work was initiated by Rubinstein, Schramm, and Weinberg (ITCS 2018), who gave a randomized algorithm that computes a minimum $s$--$t$ cut using $\widetilde{O}(n^{5/3})$ queries, thereby showing that one can avoid spending $\widetilde{\Theta}(n^2)$ queries required to learn the entire graph.\footnote{Throughout this work, we use $\widetilde{O}(\cdot)$, $\widetilde{\Theta}(\cdot)$, $\widetilde{\Omega}(\cdot)$, and $\widetilde{\omega}(\cdot)$ to hide suitable $\mathrm{polylog}(n)$ factors.} A recent result by Anand, Saranurak, and Wang (SODA 2025) also matched this upper bound via a deterministic algorithm based on blocking flows.

In this work, we present a new randomized algorithm that improves the cut-query complexity to $\widetilde{O}(n^{8/5})$. At the heart of our approach is a query-efficient subroutine that incrementally reveals the graph edge-by-edge while increasing the maximum $s$--$t$ flow in the learned subgraph at a rate faster than classical augmenting-path methods. Notably, our algorithm is simple, purely combinatorial, and can be naturally interpreted as a recursive greedy procedure.

As a further consequence, we obtain a \emph{deterministic} and \emph{combinatorial} two-party communication protocol for computing a minimum $s$--$t$ cut using $\widetilde{O}(n^{11/7})$ bits of communication. This improves upon the previous best bound of $\widetilde{O}(n^{5/3})$, which was obtained via reductions from the aforementioned cut-query algorithms. In parallel, it has been observed that an $\widetilde{O}(n^{3/2})$-bit randomized protocol can be achieved via continuous optimization techniques; however, these methods are fundamentally different from our combinatorial approach.


\end{abstract}

\thispagestyle{empty}
\newpage
\tableofcontents
\pagenumbering{roman}
\newpage
\pagenumbering{arabic}

\section{Introduction}



The \emph{minimum $s$--$t$ cut} problem is a cornerstone of combinatorial optimization and graph algorithms. Formally, given an undirected, unweighted graph $G = (V, E)$ on $n$ vertices, along with two distinguished terminals $s, t \in V$, the goal is to compute an $s$--$t$ cut of the vertex set—that is, a partition $(S, T)$ with $s \in S$ and $t \in T$—that minimizes the number of edges crossing the cut, i.e., the quantity $|E(S,T)| := |E \cap (S \times T)|$. In this work, we study the complexity of computing a minimum $s$--$t$ cut in two computational models that restrict access to the input graph.

\paragraph{Cut-query.} In this setting, we may query an oracle with a partition $(S, T)$  of $V$ (i.e. \emph{a cut}) to obtain the cut value $|E(S, T)|$. Our goal is to compute a minimum $s$-$t$ cut of $G$ while minimizing the total number of queries made to the oracle. This problem was first introduced by the work of Rubinstein, Schramm, and Weinberg \cite{RubinsteinSW18} in an effort to study submodular function minimization problem.\footnote{The problem was in fact studied earlier by Cunningham~\cite{Cunningham85} who observed that the problem can be solved in $\widetilde{O}(n^2)$ queries even on weighted directed graphs. Note, however, that the $\widetilde{O}(n^2)$ query complexity is trivial for undirected graphs.} While the ``trivial'' $O(n^2)$ queries suffices to \emph{learn} the whole undirected graph, the work of \cite{RubinsteinSW18} were the first to break this quadratic barrier via an $\widetilde{O}(n^{5/3})$-query algorithm utilizing the randomized construction of a cut sparsifier. More recently, Anand, Saranurak, and Wang \cite{AnandSW25} developed a new deterministic algorithm, matching the same $\widetilde{O}(n^{5/3})$ query complexity. Notably, their approach is purely combinatorial via an implementation of Dinitz's blocking flow algorithm \cite{Dinitz70}, thereby is vastly different from the techniques deployed in \cite{RubinsteinSW18}. Nevertheless, both algorithms circumvent the necessity to learn the whole graph which information-theoretically require $\widetilde{\Omega}(n^2)$ queries to the oracle when the graph is dense.


The works of Rubinstein, Schramm, and Weinberg \cite{RubinsteinSW18} and Anand, Saranurak, and Wang \cite{AnandSW25} approach the $s$--$t$ minimum cut problem from two completely orthogonal perspectives; yet both arrive at algorithms with the same query complexity of $\tO(n^{5/3})$. This convergence naturally raises the question: \emph{Do the current $\tO(n^{5/3})$-query upper bounds represent a latent barrier for computing a minimum $s$-$t$ cut?} In this work, we show that we can indeed go below this bound.

\begin{restatable}[Main result \#1]{them}{mincutcq}There is a randomized algorithm that computes a minimum $s$-$t$ cut of an undirected unweighted graph $G$. The algorithm always makes $\widetilde{O}(n^{8/5})$ cut queries to $G$ and succeeds with high probability.
\label{thm:min_cut_cq}
\end{restatable} 

In addition, our algorithm can be extended to enumerate \emph{all} minimum $s$--$t$ cuts within the same query complexity, although we do not make this extension explicit in our later analysis. We also highlight that our algorithms are simple, purely combinatorial, and can be naturally interpreted as a recursive greedy procedure.

\paragraph{Two-player communication.} In this setting, input graph $G = (V,E)$ is edge-partitioned across two players, namely Alice has $G_A = (V, E_A)$, and Bob has $G_B = (V, E_B)$, with $E = E_A \cup E_B$. The players are allowed to exchange a sequence of messages (via a \emph{protocol}) with the goal of computing a minimum $s$-$t$ cut of $G$ while minimizing the total amount of communication between the players.

In fact, any upper bound in the cut-query model translates to an upper bound in the communication setting, with an additional $O(\log n)$ factor. To see this, observe that any cut query can be simulated with $O(\log n)$ bits of communication: Alice locally computes $|E_A(S, T)|$ and sends the result to Bob, who then computes $|E_B(S, T)|$ and sends it back to Alice. The players can thus recover the cut value by computing $|E(S, T)| = |E_A(S, T)| + |E_B(S, T)|$. Since both $|E_A(S, T)|$ and $|E_B(S, T)|$ are bounded by $O(n^2)$, they can each be represented with $O(\log n)$ bits, making the communication cost per query $O(\log n)$. 

As an immediate corollary, simulating the deterministic $\widetilde{O}(n^{5/3})$-query algorithm of Anand, Saranurak, and Wang~\cite{AnandSW25} yields a deterministic communication protocol with total cost $\widetilde{O}(n^{5/3})$ bits. Furthermore, although the algorithm of Rubinstein, Schramm, and Weinberg~\cite{RubinsteinSW18} is originally randomized in the cut-query setting, it can be derandomized in the communication model, resulting in the same $\widetilde{O}(n^{5/3})$ deterministic upper bound. To our knowledge, prior to this work, no subquadratic deterministic or combinatorial protocol for computing a minimum $s$--$t$ cut is known beyond simulating one of these cut-query algorithms.

As a by-product of \Cref{thm:min_cut_cq}, our $\tO(n^{1.6})$-query algorithm immediately yields a two-player communication protocol with total communication cost of $\tO(n^{1.6})$ bits. Furthermore, through a careful implementation of our algorithm, we obtain a \emph{deterministic} communication protocol with the improved complexity of $\tO(n^{1.58})$ bits. This modest improvement arises from exploiting some key distinctions between the two models, thereby allowing us to realize the power of communication setting more truly.




\begin{restatable}[Main result \#2]{them}{mincutcomm} There is a deterministic two-player communication protocol that computes a minimum $s$-$t$ cut of an undirected unweighted graph $G$. The protocol uses $\widetilde{O}(n \hspace{.15mm} \nu^{4/7})$ bits of communication where $\nu$ denotes the value of $G$'s minimum $s$-$t$ cut, and always succeeds.
\label{thm:min_cut_comm}
\end{restatable} 


In parallel, it is known in the literature that one can achieve an $\tO(n^{1.5})$-bit protocol via a randomized non-combinatorial approach based on continuous optimizations—e.g., as observed by  \cite{HY25}.\footnote{The existence of a $\tO(n^{1.5})$-bit protocol was brought to our attention through private communication with experts in the area. To respect the double-blind reviewing policy, we refrain from disclosing their identities.}  Nevertheless, our result of \Cref{thm:min_cut_comm} serves as the first \emph{deterministic} and \emph{combinatorial} communication protocol for computing a minimum $s$--$t$ cut that surpasses the $\tO(n^{5/3})$ bounds.


\subsection{Related Works}

Although not the focus of our work, we briefly highlight related results on the cut-query complexity of a closely related problem: computing a \emph{global} minimum cut. The seminal work of Rubinstein, Schramm, and Weinberg~\cite{RubinsteinSW18} also initiated the study of this problem for undirected, unweighted graphs and gave a randomized $\widetilde{O}(n)$-query algorithm, which spurred a sequence of follow-up works. Apers et al.~\cite{ApersEGLMN22} subsequently removed all polylogarithmic factors, obtaining the first truly linear $O(n)$-query algorithm. More recently, Anand, Saranurak, and Wang~\cite{AnandSW25} developed the first \emph{deterministic} algorithm for this problem, using $\widetilde{O}(n^{5/3})$ queries via a novel expander decomposition technique. Kenneth-Mordoch and Krauthgamer~\cite{kenneth-mordoch2025} explored round-query trade-offs and showed that a global minimum cut can be computed in only two rounds using $\widetilde{O}(n^{4/3})$ queries. For weighted undirected graphs, Mukhopadhyay and Nanongkai~\cite{MukhopadhyayN20} achieved a near-linear $\widetilde{O}(n)$-query upper bound. Finally, the algorithm of \cite{RubinsteinSW18} was shown to be implementable in the semi-streaming model using only two passes as observed by Assadi, Chen, and Khanna~\cite{AssadiCK19}, and was later simplified by Assadi and Dudeja~\cite{AssadiD21}.

\paragraph{Connections to Submodular Function Minimization.}  
A set function $f: 2^{[n]} \to \mathbb{R}$ is called \emph{submodular} if it satisfies a diminishing returns property: for every $A \subseteq B \subseteq [n]$ and $x \notin B$, it holds that $f(A \cup \{x\}) - f(A) \ge f(B \cup \{x\}) - f(B)$. The goal of the submodular function minimization (SFM) problem is to find a minimizer of $f$, given query access to a value oracle, while minimizing the number of queries. After decades of work~\cite{GrotschelLS81, Cunningham85, Queyranne98, FleischerI00, Iwata03, Vygen03, Orlin07, IwataO09, LeeSW15, AxelrodLS20, Jiang21}, the state-of-the-art upper bound is $O(n^2 \log n)$, due to Jiang~\cite{Jiang23}. On the other hand, the current best deterministic lower bound remains at $\Omega(n \log n)$, established by Chakrabarty, Graur, Jiang, and Sidford~\cite{ChakrabartyGJS22}, following a sequence of earlier linear-query lower bounds~\cite{Harvey08, ChakrabartyLSW17, GraurPRW20}. We also highlight another recent work by Chakrabarty, Graur, Jiang, and Sidford \cite{ChakrabartyGJS23} which investigates the parallel complexity of SFM and may be of independent interest.

There is a well-known connection between SFM and the problem of computing minimum cuts in graphs. In fact, the minimum cut function---whether global or $s$-$t$, weighted or unweighted, directed or undirected---forms a canonical example of a submodular function. This connection has long motivated efforts to reduce SFM to specific cut problems. As emphasized in Rubinstein, Schramm, and Weinberg~\cite{RubinsteinSW18}, one major motivation for studying the cut-query complexity of minimum cuts was to use them as a proxy for proving quadratic lower bounds for SFM.

However, this approach faces significant limitations. The randomized $\widetilde{O}(n)$-query algorithm of~\cite{RubinsteinSW18} and the subsequent truly linear $O(n)$-query algorithm of Apers et al.~\cite{ApersEGLMN22} exclude any $\omega(n)$ lower bound for SFM via reductions from minimum cuts in unweighted, undirected graphs. Similarly, the near-linear upper bound of Mukhopadhyay and Nanongkai~\cite{MukhopadhyayN20} rules out $\widetilde{\omega}(n)$ lower bounds from weighted graphs. Our own result (\Cref{thm:min_cut_cq}) also rules out $\widetilde{\omega}(n^{8/5})$ lower bounds from minimum $s$-$t$ cuts in unweighted graphs. Nevertheless, even a modest improvement---say, proving an $\Omega(n^{1.01})$ lower bound for SFM via minimum $s$-$t$ cuts---would be a notable breakthrough.


Perhaps the most promising path toward proving quadratic lower bounds for SFM lies in the minimum cut problems on \emph{directed} graphs. The common techniques and subroutines that apply to designing cut query algorithms in undirected graphs often break down in the directed setting. At the extreme, even learning the underlying graph is impossible with only cut-query access: if a directed graph contains a cycle, reversing the orientation of the cycle leaves all directed cut values unchanged, leaving the two graphs indistinguishable via any number of cut queries. To this end, while we cannot hope to recover the entire directed input graph using only cut queries, Cunningham~\cite{Cunningham85} observed that $\widetilde{O}(n^2)$ queries suffice to compute minimum cuts in directed graphs, even when the graph is weighted. To the best of our knowledge, this quadratic upper bound has stood ever since, and not even the slightest improvement to an $n^{2 - o(1)}$ upper bound is known.\footnote{Even for \emph{reachability}, perhaps the simplest problem in directed graphs, no subquadratic-query upper bound is known.} Therefore, any progress---either an $\omega(n)$-query lower bound or an $o(n^2)$-query upper bound---would represent a major advance.

\subsection{Organization}
The remainder of this work is organized as follows. In \Cref{sec:prelim}, we introduce the notation used throughout the paper. Readers familiar with graph-theoretic concepts, the cut-query model, and the two-party communication model may skip directly to \Cref{sec:setup}, where we provide a high-level overview and sketch the main ideas behind our algorithm. The implementation and analysis in the cut-query model are presented in a combination of \Cref{sec:RSW} and \Cref{sec:large_flow_cq}. The implementation and analysis in the two-player communication model are deferred to \Cref{sec:large_flow_comm}. We conclude with a discussion of open problems in \Cref{sec:open_problems}. Finally, \Cref{appendix:flow_cover_proof} and \Cref{appendix:optimal_sparsifier} contain the missing proofs from earlier sections.


\section{Notations and Preliminaries}
\label{sec:prelim}

In this section, we introduce the terminologies and computation models that will be used throughout the paper. Due to brevity, some of the proofs are deferred to \Cref{appendix:flow_cover_proof}.


\paragraph{Graph.} Let $G = (V, E)$ denote a graph with a designated terminals $s$ and $t$.  Unless stated otherwise, a graph is assumed to be undirected and unweighted. For two graphs $G = (V, E)$ and $G' = (V,E')$ on the same vertex set, we denote $G \cup G' = (V, E \cup E')$ as the graph obtained by combining the edges of $G$ and $G'$. We also write $G \cup e$ for the addition of a single edge $e$ to $G$.

In this work, we also study graphs whose edge sets comprise both undirected unweighted edge (i.e. unit-weight) and directed weighted edges. We call a graph of this type a \emph{mixed graph}, denoted by $\mathcal{G} = (G, F)$ where $G$ is undirected unweighted graph, and $F$ is a directed graph with bounded weights. For a directed edge, we write $(a, b)$ to represent an edge from $a$ to $b$. For an undirected edge, we may write it as either $(a, b)$ or $(b, a)$ interchangeably, with no implied direction.

For a weighted graph $G = (V, E)$, we denote $w_G: E \rightarrow \mathbb{Z}_{\geq 0}$ as the weight function that maps an edge of $G$ to its positive integral weight. The total weight of any edge set $E' \subseteq E$ is given by $w_G(E') := \sum_{e \in E'} w_G(e)$. Also denote $w(G)$ or $|G|$ to be the total weight of edges in the graph.

Let $G = (V, E_G)$ and $H = (V, E_H)$ be two graphs on the same set of vertices which may be directed or weighted. We say $H$ is a \emph{subgraph} of $G$, denoted $H \preceq G$, iff $E_H \subseteq E_G$, and $w_H(e) \leq w_G(e)$ for any edge $e \in E_H$.  When both $G$ and $H$ are undirected unweighted, the second condition is not necessary.

For any graph $G = (V,E)$ and $U \subseteq V$, denote $G[U]$ to be an induced subgraph of $G$ on vertex set $U$. For a partition $(V_1,...,V_z)$ of $V$, denote $G\langle V_1,..,V_z\rangle$ to be a (weighted) graph over $z$ super-vertices resulted from contracting each vertex set $V_i$ into one super-vertex. Equivalently, we can view $G\langle V_1,..,V_z\rangle$ as $G$ with all edges within each $G[V_i]$ removed.

\paragraph{Flow and Maximum Flow.} Let $G = (V, E)$ be a (possibly directed weighted) graph with a designated source $s$ and sink $t$. A \emph{unit of $s$-$t$ flow} (or simply a \emph{unit flow}) is a directed path from $s$ to $t$. We say a directed graph $F$ is a \emph{flow graph of value $f$} if and only if $F$ can be exactly partitioned into $f$ unit flows. We denote $\maxflow(G)$ as the value of $G$'s largest flow subgraph. When $G$ is undirected unweighted, $\maxflow(G)$ corresponds to the maximum number of edge-disjoint $s$-$t$ paths that can be packed into $G$.

\begin{lemma}[Lemma 5.4 of \cite{RubinsteinSW18}]
Let $F = (V, E)$ be a non-circular flow graph of value $f$ with $n$ vertices. Let $W$ be the maximum weight of edges in $F$.  Then, we have $w(F) \leq O(n\sqrt{fW})$.  
\label{lemma:flow_cover_weighted}
\end{lemma}

The following lemma bounds the number of edges in a flow graph.

\begin{restatable}{lem}{flowcover} 
Let $F = (V, E)$ be a non-circular flow graph of value $f$ with $n$ vertices. Then, $F$ contains $|E| \leq O(n\sqrt{f})$ edges.  
\label{lem:flow_cover}
\end{restatable}

\paragraph{Residual Graph.}Given an undirected unweighted graph $G$ and a flow subgraph $F \preceq G$, the \emph{residual graph} of $G$ induced by $F$, denoted $G_F$, is obtained by replacing each edge in $F$ with its reverse edge, assigned a weight of 2. Notably, $G_F$ is a mixed graph consisting of the undirected unweighted part $G \setminus F$, and a directed weighted part which is a reversal of $F$ with weight 2 on each edge (see \Cref{fig:residual} for an illustration).

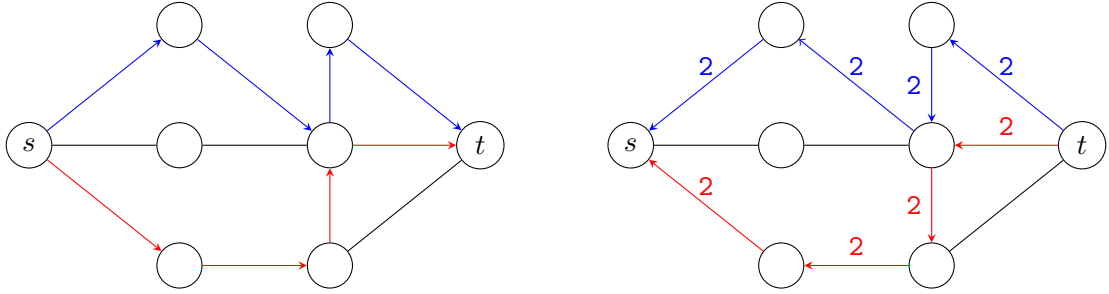
\begin{figure}[h]
    \centering
    \begin{tikzpicture}

    \node[draw, circle, minimum size = 6mm] (s) at (-2, -0.8) {$s$};
    
    \node[draw, circle, minimum size = 6mm] (A1) at (0, 0.8) {};
    \node[draw, circle, minimum size = 6mm] (A2) at (0, -0.8) {};
    \node[draw, circle, minimum size = 6mm] (A3) at (0, -2.4) {};
    
    \node[draw, circle, minimum size = 6mm] (B1) at (2, 0.8) {};
    \node[draw, circle, minimum size = 6mm] (B2) at (2, -0.8) {};
    \node[draw, circle, minimum size = 6mm] (B3) at (2, -2.4) {};

    \node[draw, circle, minimum size = 6mm] (t) at (4, -0.8) {$t$};

    \draw[->,  >=stealth, blue] (s) -- (A1);
    \draw (s) -- (A2);
    \draw[->, >=stealth, red]  (s) -- (A3);

    \draw[->,  >=stealth, blue] (A1) -- (B2);
    \draw (A2) -- (B2);
    \draw[->,  >=stealth, red]  (A3) -- (B3);

    \draw[->,  >=stealth, blue] (B1) -- (t);
    \draw[->,  >=stealth, red] (B2) -- (t);
    \draw[->,  >=stealth, blue] (B2) -- (B1);
    \draw[->,  >=stealth, red] (B3) -- (B2);
    \draw (B3) -- (t);

    \node[draw, circle, minimum size = 6mm] (s') at (6, -0.8) {$s$};
    
    \node[draw, circle, minimum size = 6mm] (A1') at (8, 0.8) {};
    \node[draw, circle, minimum size = 6mm] (A2') at (8, -0.8) {};
    \node[draw, circle, minimum size = 6mm] (A3') at (8, -2.4) {};
    
    \node[draw, circle, minimum size = 6mm] (B1') at (10, 0.8) {};
    \node[draw, circle, minimum size = 6mm] (B2') at (10, -0.8) {};
    \node[draw, circle, minimum size = 6mm] (B3') at (10, -2.4) {};

    \node[draw, circle, minimum size = 6mm] (t') at (12, -0.8) {$t$};

    \draw[<-, >=stealth, blue] (s') -- node[midway, above] {\tt 2} (A1');
    \draw (s') -- (A2');
    \draw[<-, >=stealth, red] (s') -- node[midway, above ] {\tt 2}(A3');

    \draw[<-, blue] (A1') -- node[midway, above] {\tt 2}  (B2');
    \draw (A2') -- (B2');
    \draw[<-,  >=stealth, red] (A3') -- node[midway, above] {\tt 2}(B3');

    \draw[<-, >=stealth,blue] (B1') -- node[midway, above] {\tt 2}  (t');
    \draw[<-, >=stealth,red] (B2') -- node[midway, above] {\tt 2}  (t');
    \draw[<-, >=stealth,blue] (B2') -- node[midway, left] {\tt 2}  (B1');
    \draw[<-, >=stealth,red] (B3') -- node[midway, left] {\tt 2}  (B2');
    \draw (B3') -- (t');  

\end{tikzpicture}
    \caption{Shown on the left is a graph $G$ with a flow graph $F$ of size 2 indicated by the union of red and blue edges. Shown on the right is $G_F$ which is the residual graph of $G$ induced by $F$.}
    \label{fig:residual}
\end{figure}

\begin{proposition} For any graph $G$ and a flow subgraph $F \preceq G$, we have:
$$\maxflow(F) + \maxflow(G_F) = \maxflow(G).$$
\label{prop:sum_flows}
\end{proposition}

\paragraph{Cut and Minimum Cut.} Let $G = (V, E)$ be a (possibly directed weighted) graph with a designated terminals $s$ and $t$. A \emph{cut} of $G$ is defined as a partition $(S,T)$ of the vertex set $V$. If $s \in S$ and $t \in T$, we say such cut is an \emph{$s$-$t$ cut}. An undirected/directed edge $(a,b) \in E$ is said to \emph{cross} the cut if $a \in S$ and $b \in T$. The set of edges crossing the cut $(S,T)$ is denoted by $E_G(S, T)$. The \emph{size} of a cut $(S, T)$ is defined as the cardinality of $E_G(S, T)$, given by $|E_G(S, T)|$. The \emph{value} of a cut is the total weight of the edges crossing the cut, quantified as $w_G(S, T) = \sum_{e \in E_G(S,T)} w_G(e)$. When $G$ is undirected and unweighted, we may write the cut value by $w_G(S,T)$ or $|E_G(S,T)|$. For clarity, we reserve the notation of cut value $|E_G(S,T)|$ solely for the case that $G$ is undirected and unweighted. For the remaining cases, we write the cut value as $w_G(S,T)$.

The \emph{minimum $s$-$t$ cut value} of $G$, denoted by $\mincut(G)$, is the minimum value among all possible $s$-$t$ cuts. An $s$-$t$ cut that attain the minimum cut value is called a \emph{minimum $s$-$t$ cut}. A cut (not necessarily an $s$-$t$ cut) of smallest value is called a \emph{global minimum cut}. The following theorem, established in \Cite{Ford56}, formalizes the duality between maximum flow and minimum cuts.

\begin{theorem}[Max-Flow Min-Cut Theorem]
For any graph $G$ with a source $s$ and sink $t$, the maximum value of an $s$-$t$ flow equals the minimum value of an $s$-$t$ cut, formally: $$\maxflow(G) = \mincut(G).$$

\label{thm:maxflow_mincut}
\end{theorem}



Throughout this paper, we use $\nu(G)$ to denote the value of the maximum $s$–$t$ flow and, equivalently, the value of the minimum $s$–$t$ cut in the graph $G$:
$$\nu(G) := \maxflow(G) = \mincut(G).$$
More precisely, we define $\nu_{s,t}(G)$ as the value of the minimum $s$–$t$ cut. Since we always fix a specific pair of terminals $(s, t)$, we omit the subscripts and simply write $\nu(G)$.

\paragraph{$k$-Connected Component.} Let $G = (V,E)$ be an undirected graph that can possibly be weighted. We say $U \subseteq V$ is \emph{$k$-connected component} iff $G[U]$ has \emph{global} min-cut of value at least $k$. The following lemma is known. For its proof, we refer readers to Lemma 3.8 of \cite{RubinsteinSW18}.

\begin{lemma} Let $G$ be an $n$-vertex undirected weighted graph with no $k$-connected components. Then, we have $w(G) \leq k(n-1)$.
\label{lem:small_connectivity_edge_bounds}
\end{lemma}

\paragraph{Cut Sparsifier.} Let $G$ be an undirected graph that may be weighted, and let $\eps \in (0,1)$. We say an undirected weighted graph $H$ is an $\eps$-\emph{cut sparsifier} of $G$ iff $E(H) \subseteq E(G)$ and we have $(1-\eps) \cdot w_G(S,T) \leq w_H(S,T) \leq (1+\eps) \cdot w_G(S,T)$ for any cut $(S,T)$. 

It is widely known that for any $\eps \in (0,1)$, any graph $G$ has an $\eps$-cut sparsifier with $\tO(n/\eps^2)$ edges. We state the version of sparisifer we need via the following theorem. 

\begin{restatable}{them}{sparsifiersmallweight} Let $G$ be an unweighted undirected connected graph with $n$ vertices and average degree $d$. Then, for any $\eps \in (0,1)$, there exists an $\eps$-cut sparsifier of $G$ with $O\left(\frac{n \log^2{n}}{\eps^2}\right)$ edges and weights bounded by $O\left(\frac{\eps^2d}{\log^2{n}}\right)$.
\label{lem:sparsifier_small_weights}
\end{restatable}

We make several remarks. First, the theorem establishes an asymptotically optimal trade-off between the number of edges and the maximum weight of the sparsifier. To see this, let $G$ be a graph with $m$ edges. Any $\varepsilon$-cut sparsifier must have total weight at least $(1 - \varepsilon)m = \Omega(m)$. Since our sparsifier contains $O\left(\frac{n \log^2 n}{\varepsilon^2}\right)$ edges, at least one edge must carry a weight $\frac{\Omega(m)}{O\left(\frac{n \log^2 n}{\varepsilon^2}\right)} = \Omega\left(\frac{\varepsilon^2 d}{ \log^2 n}\right)$,
which matches our upper bound up to constant factors.

Second, for completeness, the proof of \Cref{lem:sparsifier_small_weights} is deferred to \Cref{appendix:optimal_sparsifier}. In fact, we prove a stronger statement for an extended notion known as a \emph{spectral} sparsifier. We emphasize that our proof does not claim novelty in technique. Many foundational results on sparsification are well established \cite{Karger99, BenczurK15, FungHHP19, SpielmanS11, SpielmanT11, BatsonSS12, ChuGPSSW23, Lau0025}, and our argument merely follows prior approaches, such as that of \cite{ChuGPSSW23}. However, to the best of our knowledge, previous work has not explicitly addressed bounding the maximum edge weight in sparsifiers. When bounds are provided, they are typically of order $\tO(\varepsilon^2 n)$. Our contribution lies in explicitly deriving \emph{asymptotically optimal} bounds on the maximum edge weight. This refinement turns out to be critical for our algorithmic applications.

\subsection{Cut-Query Model}

In this setting, an input graph $G = (V,E)$ is hidden from us, and we are only allow to make \emph{cut-queries} to the graph. In particular, for each query we can specify a cut $(S,T)$ to an oracle, and receives $|E_G(S,T)|$ in return. Our objective is to compute a task (in this case, a minimum $s$-$t$ cut) while minimizing the query complexity of our algorithm. Note that the only notion of cost here is the queries; thus, we can assume to have an unlimited local computation power.

Many fundamental tasks can be simulated efficiently via cut queries. For instance, we can find out if an edge $(a,b)$ exists in the graph within $O(1)$ queries, learn all the $m$ edges in the graph within $O(n + m \log n)$ queries, or uniformly sample an edge of the graph with an amortized cost of $O(\log n)$ queries. This random sampling primitive also opens up many doors of opportunity to the world of sequential algorithm. Perhaps one of the most insightful takeaways from \cite{RubinsteinSW18} is that we can obtain a cut sparsifier of $G$ within not-too-many cut queries to $G$, thereby enabling us to obtain a $(1\pm \eps)$-approximation of the \emph{value} of the minimum $s$-$t$ cut.

\begin{lemma}[Theorem 3.5 of \cite{RubinsteinSW18}]  For any $\eps \in (0,1)$, there is a randomized algorithm that computes an $\eps$-cut sparsifier of $G$ with $\tO(n/\eps^2)$ edges and maximum weight $\tO(\eps^2 n)$. The algorithm uses $\tO(n/\eps^2)$ cut queries to $G$ and succeeds w.h.p.
\label{lem:sparsifier_cq}
\end{lemma} 


\begin{restatable}{coro}{approxnucq} For any $\eps \in (0,1)$, there is a randomized algorithm that computes $f$ such that $(1-\eps) \nu(G) \leq f \leq (1+\eps) \nu(G)$. The algorithm uses $\tO(n/\eps^2)$ cut queries to $G$ and succeeds w.h.p.
\label{cor:approx_nu_cq}
\end{restatable}

Throughout this paper, we will utilize a more specific type of query to a graph $G = (V,E)$ called the \emph{find-edge query} which can be simulated efficiently via cut queries.

\begin{restatable}[Find-edge query]{defi}{ISSdef} Let $G = (V,E)$ be an underlying graph. A \emph{find-edge} query takes in as inputs disjoint sets $A,B \subseteq V$, and outputs $\FindEdge(A,B)$ which is either \textnormal{``NONE''} if $E \cap (A \times B) = \emptyset$, or an edge $(a,b) \in E \cap (A \times B)$ if otherwise.
\end{restatable}

\begin{restatable}{prop}{ISScq} Let $G = (V,E)$ be an underlying graph. Then, for any disjoint sets $A,B \subseteq V$, we can determine $\FindEdge(A,B)$ using either
\begin{itemize}
    \item $O(1)$ cut queries to $G$ if $\FindEdge(A,B)$ outputs \textnormal{``NONE''} (i.e. $E \cap (A \times B) = \emptyset$), or
    \item $O(\log n)$ cut queries to $G$ if $\FindEdge(A,B) $ outputs an edge $(a,b) \in E \cap (A \times B)$. 
\end{itemize}
\label{prop:ISS_cq}
\end{restatable}

One of its abundant benefits of the find-edge queries is that it allows us to learn the (contracted) graphs with query complexity near-linear in the representation size of the graph.

\begin{restatable}{lem}{learngraphcq} Let $(V_1,...,V_z)$ be a partition of a vertex set of $G = (V,E)$. Then, we can learn all edges of $G\langle V_1,...,V_z \rangle$ using $ \tO\left(z + \left|G\langle V_1,...,V_z \rangle\right| \right)$ cut queries.
\label{lem:learn_graph_cq}  
\end{restatable}

\subsection{Two-Party Communication Model}

In this setting, an input graph $G = (V,E)$ is edge-partitioned across two players, namely Alice has her private graph $G_A = (V, E_A)$, and Bob has his private graph $G_B = (V, E_B)$, with $E = E_A \cup E_B$. The players' objective is to collaboratively compute a joint task (in this case, a minimum $s$-$t$ cut) by exchanging a sequence of messages, known as a \emph{protocol}, while minimizing the total amount of communication.  Note that the only notion of cost here is the communication; thus, Alice and Bob can be assumed to have unlimited local computation power. 

Within this communication setting, the players can also ``deterministically'' compute a cut sparsifier and use it to approximate the max-flow/min-cut value with the same complexity.

\begin{restatable}{lem}{sparsifiercomm} For any $\eps \in (0,1)$, there is a deterministic protocol such that Alice and Bob jointly computes an $\eps$-cut sparsifier of $G$ with $\tO(n/\eps^2)$ edges and maximum weight $\tO(\eps^2 n)$. The protocol communicates $\tO(n/\eps^2)$ bits and always succeeds.
\label{lem:sparsifier_comm}
\end{restatable}

\begin{corollary} For any $\eps \in (0,1)$, Alice and Bob can jointly and deterministically compute a value $f$ such that $(1-\eps) \nu(G) \leq f \leq (1+\eps) \nu(G)$ using $\tO(n/\eps^2)$ bits of communication. 
\label{cor:approx_nu_comm}
\end{corollary}


\section{Overview}
\label{sec:setup}

The starting point of our algorithm builds on the work of \cite{RubinsteinSW18}. Their algorithm, in which we will call RSW algorithm, can be viewed as computing a partition of the vertex set whose contraction results in a sparse graph that preserves all $s$-$t$ minimum cuts. That work focuses exclusively on the cut-query model and yields a randomized algorithm with $\tO(n^{5/3})$ query complexity for the task. Upon closer inspection, we observe that the complexity can be parameterized in terms of $\nu(G)$ depending on the computation models: to $\tO(n^{4/3} \nu(G)^{1/3})$ cut queries and to $\tO(n \cdot \nu(G)^{2/3})$ bits of communication. At first glance, these parameterized improvements may seem unhelpful, as $\nu(G)$ can be as large as $\Omega(n)$. However, to state the obvious, when $\nu(G)$ is truly sublinear, these bounds yield a complexity of $n^{5/3 - \Omega(1)}$ in both models. This observation motivates a simple yet powerful scheme: apply a preprocessing step to $G$ that reduces the $s$-$t$ min-cut to $\Delta = o(n)$ while preserving all $s$-$t$ min-cuts, and then apply the RSW algorithm to the post-processed graph. If this preprocessing can be implemented efficiently, it opens the possibility of achieving overall sub-$n^{5/3}$ complexity. We summarize this framework in the table below.

\renewcommand{\tablename}{Algorithm}

\begin{table}[H]
    \centering
    \begin{tabular}{|p{15.5cm}|}
    \hline ~\\
   \multicolumn{1}{|c|}
     {} \\
    \vspace{-7.5mm} 
    \textbf{Algorithm 1:} A model-independent framework for computing minimum $s$-$t$ cuts. \\

    \textbf{Input:}  An undirected unweighted graph $G = (V,E)$ with $n$ vertices. \\
    \textbf{Output:} A minimum $s$-$t$ cut of $G$. \\~\\

    \textbf{Procedures:}
    \begin{enumerate}
        \item Compute a graph $H \preceq G$ such that $\nu(H) \geq \nu(G) - \Delta$ for some value $\Delta$ to be set later.
        \begin{itemize}
            \item[1a.] Compute a flow subgraph $F \preceq H$ of value $\nu(F) = \nu(H) \geq \nu(G) - \Delta$.
        \end{itemize} 
    \item Compute a minimum $s$-$t$ cut in the residual network $G_F$, and output it.
    \end{enumerate}
        \vspace{0mm}
    \\
    \hline
    \end{tabular}
    \caption{A model-independent framework for computing a minimum $s$-$t$ cut.}
    \label{alg:meta-alg}
\end{table}

Note here that if we assume a correct answer $H$ following Step 1, the flow subgraph $F$ in Step 1a can be computed locally and free of costs in both computation models.

We first need to argue that this is a correct framework. This is done via following claim which gives a perfect reduction between all minimum $s$-$t$ cuts of $G_F$ and all minimum $s$-$t$ cuts of $G$.

\begin{claim} Let $G$ be an undirected unweighted graph and $F$ be its flow subgraph. For any cut $(S,T)$, we have $|E_G(S,T)| = \nu(F) + w_{G_F}(S,T).$ 
\label{clm:reduce_cut_by_flow}
\end{claim}
\begin{proof} First, observe that $w_{G_F}(S,T) = |E_{G \setminus F}(S,T)| + 2 \cdot w_F(T,S)$. This is because we can write $G_F$ as a union of two graphs: an undirected unweighted graph $G \setminus F$ (which contributes to the first term) and a directed graph which is a reversal of edges of $F$ with weights 2 (which contributes to the second term).

Observe further that $\nu(F) = w_F(S,T) - w_F(T,S)$. This is because the edges of $F$ can be partitioned into $\nu(F)$ disjoint $s$-$t$ directed paths. Each of those start at $s \in S$ and ends at $t \in T$. As such, it must cross from $S$ to $T$ exactly one more time than it crosses from $T$ to $S$. Thus, each of the $\nu(F)$ disjoint $s$-$t$ directed paths contribute a value of $1$ to $w_F(S,T) - w_F(T,S)$. 

Combining the two observations, we have: 
\begin{align*}
\nu(F) + |E_{G_F}(S,T)|  = w_F(S,T) + w_F(T,S) + |E_{G \setminus F}(S,T)| = |E_G(S,T)|
\end{align*}
as wished.
\end{proof}

We consider the implementation of Algorithm~\ref{alg:meta-alg} in two computational models: the \emph{cut-query model} and the \emph{two-player communication model}, where we aim to use the (modified) RSW algorithm as a baseline to implement Step~2 in both settings. However, applying the RSW algorithm to the residual graph $G_F$ introduces two key challenges. First, we only have cut-query access to the original graph $G$, not to the residual graph $G_F$. In fact, throughout our algorithm, we shall only need to make cut queries to the graph $G\setminus F$. Fortunately, since the flow graph $F$ is explicitly computed in Step~1a, we can simulate cut queries to $G\setminus F$ using a single query access to $G$ and the explicit description of $F$. Second, the residual graph $G_F$ is a \emph{mixed} graph, containing both directed and undirected edges. As the RSW algorithm was originally designed for undirected graphs, it does not directly apply to this setting. To address this, we show that the RSW algorithm can be carefully modified to operate on mixed graphs in both computational models without altering its correctness. These modifications also allow for a query and communication complexity to be parameterized by the min-cut/max-flow value of the residual graph $G_F$ as well.

For now, we summarize the results of applying the modified RSW algorithm in both computational models with the full implementation details deferred to \Cref{sec:RSW}.

\begin{restatable}[Modified RSW via cut query]{lem}{RSWcutquery}
 Let $\mathcal{G} = (G, F)$ be an input mixed graph where $G$ is undirected and unweighted, and $F$ be a directed bounded-weighted that is known explicitly. Then, there exists an algorithm that makes $\tO(n^{4/3} \nu(\mathcal{G})^{1/3})$ cut-queries to $G$ and w.h.p. computes a minimum $s$-$t$ cut of $\mathcal{G}$.
\label{lem:RSW_cut_query}
\end{restatable}

\begin{restatable}[Modified RSW via communication]{lem}{RSWcomm}
 Let $\mathcal{G} = (G, F)$ be an input mixed graph where $G$ is undirected and unweighted and is partitioned among two players, and $F$ be a directed bounded-weighted that is known explicitly. Then, there exists an $\tO(n \cdot \nu(\mathcal{G})^{2/3})$-bit deterministic protocol that Alice and Bob jointly outputs a minimum $s$-$t$ cut of $\mathcal{G}$.
\label{lem:RSW_comm}
\end{restatable}

We now restrict our attention to Step 1 as it is key algorithmic design of this work. We first state the results on both computation models.

\begin{restatable}[Finding a large flow via cut query]{lem}{largeflowcutquery} Let $G$ be an input graph, and let $\Delta \leq n$ be an additive error parameter. Then, there is a randomized algorithm that computes $H \preceq G$ with $\nu(H) \geq \nu(G) - \Delta$. The algorithm makes $\widetilde{O}\left(\frac{n^2}{\sqrt{\Delta}} + \frac{n^3}{\Delta^2}\right)$ cut queries to $G$ and succeeds w.h.p. 
\label{lem:large_flow_qc}
\end{restatable}

\begin{restatable}[Finding a large flow via communication]{lem}{largeflowcomm} Let $G$ be an input graph, and let $\Delta \leq n$ be an additive error parameter. Then, there is a deterministic protocol that computes $H \preceq G$ with $\nu(H) \geq \nu(G) - \Delta$. The protocol uses $\widetilde{O}\left(\frac{n \cdot \nu(G)}{\sqrt{\Delta}} + \frac{n \cdot \nu(G)^2}{\Delta^2}\right)$ bits of communication and always succeeds.
\label{lem:large_flow_comm}
\end{restatable}

Our algorithms for both lemmas share the same principle. To best understand it, we study a simpified question in a model-independent setting: suppose we are given three parameters $k' < k \leq f$ with a promise that $\nu(G) \geq f$, and also given for free is a graph $H \preceq G$ such that $\nu(H) \geq f-k$. Our goal is to sequentially add (a small number of) edges from $G \setminus H$ to $H$ so that it eventually becomes $H^{\text{after}}$ with $\nu(H^{\text{after}}) \geq f-k'$. In later sections, we show that an efficient algorithm for this toy question leads to both \Cref{lem:large_flow_qc} and \Cref{lem:large_flow_comm}.

At a high level, our algorithm proceeds in a greedy-like manner. For any graph $\mathcal{G}$, we define a set of \emph{witnesses}, denoted $\mathcal{W}(\mathcal{G})$, with an invariant that $\nu(\mathcal{G}) \geq f - k'$ if and only if $\mathcal{W}(\mathcal{G}) = \emptyset$. Our procedure starts with $H_1 \leftarrow H$. In each round $i = 1, 2, \ldots$, we determine an edge $e_i \in G \setminus H_i$ that reduces the number of witnesses by a non-trivial fraction; that is, $e_i$ is such that $|\mathcal{W}(H_i \cup \{e_i\})| \leq (1-\gamma) \cdot |\mathcal{W}(H_i)|$ for some fixed value $\gamma $ that may depend on $n,k,$ or $k'$. We call such edges with that property \emph{$\gamma$-good.} We then set $H_{i+1} \leftarrow H_i \cup \{e_i\}$, and repeat. The procedure terminates at the earliest round $T$ where $\mathcal{W}(H_T) = \emptyset$, ensuring that the final graph $H^{\text{after}} := H_T$ has $\nu(H_T) \geq f - k'$ via our invariant.

Following this greedy-like approach, there are two main questions that need to be addressed. First, what is the correct definition of a \emph{witness}? Recall that we want to maintain the invariant that $\nu(\mathcal{G}) \geq f - k'$ if and only if $\mathcal{W}(\mathcal{G}) = \emptyset$. A strong candidate for a witness is a cut of size at most $f - k' - 1$ in $\mathcal{G}$, thanks to the Max-Flow Min-Cut Theorem. This, however, is only partially correct. We define a witness as a \emph{redundant} representation of a small cut on $\mathcal{G}$, consisting of a cut $(S,T)$ of small size $|E_\mathcal{G}(S,T)| \leq f - k' - 1$, the set $E_\mathcal{G}(S,T)$ of edges crossing such cut, and a set $X$ of \emph{non-existing edges} that cross this cut, with $|X| = f - k' - 1 - |E_\mathcal{G}(S,T)| \geq 0$. With this notion, we still maintain the invariant, while the added redundancy (i.e., the non-existing edges) allows us to pre-allocate a set of edges that could potentially be added to $\mathcal{G}$ in the future, ensuring that the cut value of $(S, T)$ remains at most $f - k' - 1$.

\begin{restatable}[Witness]{defi}{witness}
Let $\mathcal{G} = (V, E)$ be an unweighted undirected graph with two designated terminals $s$ and $t$. A tuple $\mathsf{W} = (S, T, Y)$ is called a \emph{witness} of $\mathcal{G}$ if the following conditions hold:  

\begin{enumerate}  
    \item $(S, T)$ is an $s$-$t$ cut of $\mathcal{G}$ with value $|E_\mathcal{G}(S, T)| \leq f - k' - 1$.  
    \item $Y \subseteq S \times T$ such that $|Y| = f-k'-1$ and $E_\mathcal{G}(S, T) \subseteq Y$.
\end{enumerate}
Finally, let $\mathcal{W}(\mathcal{G})$ denote the set of all witnesses of $\mathcal{G}$.
\label{witness}
\end{restatable}

Notice that the notion of $Y$ here is consistent with $E_\mathcal{G}(S, T) \cup X$ as explained in the previous paragraph. For clarity, we may relax the notation and write a witness as $(S, T, E_\mathcal{G}(S, T) \cup X)$ if needed. However, it is crucial to view $Y$ as a union of $E_\mathcal{G}(S, T)$ and $X$ rather than two individual sets $E_\mathcal{G}(S, T)$ and $X$.

\begin{restatable}[Witness Invariant]{inv}{invariant} For any graph $\mathcal{G}$, we have $\nu(\mathcal{G}) \geq f-k'$ if and only if $\mathcal{W}(\mathcal{G}) = \emptyset$.
\label{invariant}
\end{restatable}

The second question is how quickly we can decrease the number of witnesses. Following the above definition of witnesses, we can naturally describe how a witness is removed upon an insertion of $(a,b)$: say a pair of vertices $(a, b) \notin E$ \emph{kills} $\mathsf{W}$ iff $\mathsf{W}$ is \emph{not} a witness of $\mathcal{G} \cup (a,b)$. \footnote{Notice that this is defined with respect to a \emph{pair} of vertices $(a,b)$ rather than an \emph{edge} $(a,b)$. This is due to a technical reason that in the cut-query model, we (as an algorithm designer) do not know if $(a,b)$ exists in the underlying graph until we somehow ``learn'' it. For the ease of understanding, it is helpful to think of $(a,b)$ as an actual edge in $G$ that is not presented in $H$ yet.}

Following this description, observe that if $a \notin S$ or $b \notin T$,  then $(a,b)$ can never kill $\mathsf{W}$. Furthermore, even if $a \in S$ and $b \in T$ but $(a,b) \in Y$, then $\mathsf{W}$ remains a witness of $\mathcal{G} \cup (a,b)$. This leads us to the following equivalent interpretation of invalidating a witness.

\begin{restatable}[Killing a witness]{defi}{kill}Let $\mathcal{G} = (V, E)$ be an unweighted and undirected graph with terminals $s$ and $t$, and let $\mathsf{W} = (S, T, Y)$ be one of its witnesses. We say a pair of vertices $(a,b)$ \emph{kills} witness $\mathsf{W} = (S,T,Y)$ if and only if $a \in S$, $b \in T$, and $(a, b) \notin Y$. 
\label{def:kill}
\end{restatable}
In other words, these conditions together mean that $(a,b)$ are on the different side of $(S,T)$ and does \emph{not} belong to the pre-allocated cut $Y$ of $(S, T)$, thereby invalidating the witness.

With this definition, we can establish an upper bound $n^{O(n)}$ on the initial number of witnesses. To see this, there are up to $2^n$ possible choices of $(S,T)$, and for each fixed value of $(S,T)$, there are up to $n^{2(f-k'-1)} \leq n^{2n}$ choices of $Y$. Due to our greedy-like approach, we assert that each iteration reduces the number of witnesses by a fraction of $\gamma$. Therefore, we only need $\widetilde{O}(n/\gamma)$ iterations to eliminate all witnesses, thereby terminating the algorithm.

Finally, for this argument to go through, we need to determine an appropriate value of $\gamma$ such that a $\gamma$-good edge always exists in $G \setminus H_i$ in every iteration $i$. Depending on the computation model, we also need to be able to find such edge efficiently. Without these guarantees, our algorithm may spend too much resource in finding a $\gamma$-good edge, or even worse may fail horribly if such edge did not exists at all. This is where the ``promise'' that $\nu(G) \geq f$ and the ``redundancy'' of witness's definition comes into rescue. Under these assumptions, by setting $\gamma = \Omega\left(\frac{k'}{n\sqrt{k}}\right)$, a $\gamma$-good edge always exists.

\begin{claim} For any iteration $i$ such that $\mathcal{W}(H_i) \ne \emptyset$, there exists an edge $e_i \in G \setminus H_i$ such that 
$$|\mathcal{W}(H_i \cup  e_i)| \leq \left(1-\Omega\left(\frac{k'}{n\sqrt{k}}\right)\right) \cdot |\mathcal{W}(H_i)|.$$ 
\label{clm:good_edge_exists}
\end{claim}

For now, we will only sketch the proof of \Cref{clm:good_edge_exists}. Later down the line, we will give a constructive variant of this claim which allows us to efficiently recover such edge $e_i$ via cut-queries.
\\

\proofsketch{Let $Q \subseteq G \setminus H_i$ be such that $|Q| = O(n\sqrt{k})$ and $\nu(H_i \cup Q) \geq f$. Such set $Q$ exists by construction: take a flow subgraph $F \preceq H_i$ of size $f-k$ and let $Q$ be the set of all edges in a non-circular flow of size $k$ of $G_F$ that has not appeared in $H_i$. The bound of $|Q| = O(n\sqrt{k})$ is obtained via \Cref{lem:flow_cover}.

Suppose we were to add $Q$ to $H_i$. For any witness $\mathsf{W} = (S,T,Y)$ of $H_i$, there will be at least $f$ edges among $H_i \cup Q$ that crosses the cut $(S,T)$. However, we only pre-allocate $f-k'-1$ edges to $Y$ and those already include $E_{H_i}(S,T)$. Thus, there are at least $k'+1$ edges among $Q$ that kills $\mathsf{W}$. Since this holds for any witness $\mathsf{W}$, via an averaging argument, some edge(s) among $Q$ must kill at least a  $\frac{k'+1}{|Q|} \geq \Omega\left(\frac{k'}{n \sqrt k}\right)$ fraction of $\W(H_i)$.}

\paragraph{\Cref{clm:good_edge_exists} is good (enough) for two-player communication.} When edges of $G$ are partitioned across Alice and Bob, each player can locally evaluate, via brute-force, whether each edge is $\gamma$-good w.r.t. $\gamma = \Omega\left(\frac{k'}{n \sqrt{k}}\right)$. \Cref{clm:good_edge_exists} guarantees that such good edge always exists among $G \setminus H$; thus must be found by one of the players. Such player then spends $\widetilde{O}(1)$ bits of communication to send over that edge, adds it to $H$, and then proceeds to the next iteration. 
As a result, the total communication cost is $\tO(n/\gamma) = \tO(n^2\cdot \frac{\sqrt{k}}{k'}).$ Then, following a line of reduction, we can recover a weaker version of \Cref{lem:large_flow_comm} with communication cost replaced by $\tO\left(\frac{n^2}{\sqrt{\Delta}} + \frac{n^3}{\Delta^2}\right)$. For all intents and purposes, we can use this weaker result to recover an $\tO(n^{11/7})$ communication protocol by combining it with \Cref{lem:RSW_comm}. However, to get a parameterized  complexity of $\tO(n\cdot \nu(G)^{4/7})$ as promised in \Cref{lem:large_flow_comm}, we need to strengthen our strategy. We defer this improvement to \Cref{sec:large_flow_comm}.

\paragraph{It is a bit more complicated for cut-query.} In the communication model, our search for a $\gamma$-good edge leverages a crucial feature: the input graph is \emph{explicitly} given, although being split between players. This guarantees that any $\gamma$-good edge must already be visible to at least one of them. In contrast, the cut-query model lacks this advantage, as we can only make queries to a \emph{hidden} graph.  Given the somewhat-artificial definition of a $\gamma$-good edge, it is not immediately clear whether such an edge can be efficiently identified. Fortunately, we develop an alternative proof of \Cref{clm:good_edge_exists} that is constructive! This new approach—combined with a few additional subroutines—allows us to locate an $\Omega(\gamma)$-good edge (with a slightly worse constant factor) using only $O(\log n)$ cut queries to $G$. We explore this approach in detail in \Cref{sec:large_flow_cq}.

\subsection{Proofs of Main Theorems}

We wrap up this overview by combining the main lemmas to prove our main results.

\mincutcq*

\begin{proof} We follow the framework of Algorithm \ref{alg:meta-alg} with $\Delta = n^{4/5}$.  In addition, we use \Cref{lem:RSW_cut_query} for Step 1 and \Cref{lem:large_flow_qc} for Step 2. Assume that all steps succeed which occurs w.h.p. Note that given $H$ from Step 1, we can locally compute $F$ at no extra costs. Plus, we have $\nu(G_F) = \nu(G) - \nu(F) \leq \Delta$ via \Cref{prop:sum_flows}. The total number of queries then becomes 
$$\tO\left(\frac{n^2}{\sqrt{\Delta}} + \frac{n^3}{\Delta^2} + n^{4/3} \Delta^{1/3}\right) = \tO\left(n^{8/5}\right)$$ as $\Delta = n^{4/5}$.
\end{proof}

\mincutcomm*

\begin{proof} We first compute $\hat{\nu}$ which is a constant approximation of $\nu$ using $\tO(n)$ communication via \Cref{cor:approx_nu_comm}. Then, we follow the framework of Algorithm \ref{alg:meta-alg} using $\Delta = \hat{\nu}^{6/7}$.  In addition, we use \Cref{lem:large_flow_comm} for Step 1 and \Cref{lem:RSW_comm} for Step 2, both of which are deterministic. Note that given $H$ from Step 1, both players can locally compute $F$ at no extra costs. Plus, we have $\nu(G_F) = \nu(G) - \nu(F) \leq \Delta$ via \Cref{prop:sum_flows}. The total number of queries then becomes 
$$\tO\left(n + \frac{n \cdot \nu(G)}{\sqrt{\Delta}} + \frac{n \cdot \nu(G)^2}{\Delta^2} + n \Delta^{2/3}\right) = \tO\left(n \hspace{.15mm} \nu^{4/7}\right)$$ as $\Delta = \hat{\nu}^{6/7} = \Theta\left(\nu^{6/7}\right)$.
\end{proof}


\section{Modifications to RSW Algorithm}
\label{sec:RSW}

In this section, we discuss Step 2 of our meta algorithm. To this end, we propose a baseline procedure—independent of the computational model—that is adapted from the work of \cite{RubinsteinSW18} to handle a mixed input graph $\mathcal{G} = (G, F)$. As a reminder, in our setting the directed weighted graph $F$ is provided explicitly and at no cost; hence, the ``true'' input is an undirected, unweighted graph $G$.


\renewcommand{\tablename}{Algorithm}

\begin{table}[H]
    \centering
    \begin{tabular}{|p{15.5cm}|}
    \hline ~\\
    \multicolumn{1}{|c|}
     {} \\
    \vspace{-7.5mm} 
    \textbf{Algorithm 2:} Minimum $s$-$t$ cut algorithm for mixed graphs. \\

    \textbf{Input:}  An $n$-vertex mixed graph $\mathcal{G} = (G,F)$ with $F$ being given explicitly. \\
    \textbf{Output:} A minimum $s$-$t$ cut of $\mathcal{G}$. \\~\\

    \textbf{Procedures:}

    \begin{enumerate}
    \item Fix a value of $\varepsilon \in (0,\frac{1}{5})$.
    \item Compute $G'$ which is a $\eps$-cut sparsifier of $G$.
    \item Denote $H = G' \cup F$ and compute $\nu(H)$.
    \item Compute a maximum $s$-$t$ flow of $H$. Denote such flow graph by $\mathcal{F}$.
    \item Obtain $K = H \setminus (\mathcal{F} \cup F)$. 
    \item Compute a partition $V = V_1 \cup \ldots \cup V_z$ with the following properties.
    \begin{itemize}
        \item Each $K[V_i]$ has connectivity $\geq 3 \eps \cdot \nu(H)$.
        \item $K\langle V_1,..,V_z\rangle$ has connectivity $\leq 3\eps \cdot \nu(H)$.
    \end{itemize}
    \item Learn all edges of $G\langle V_1,..,V_z\rangle$.
    \item Output a minimum $s$-$t$ cut in $\mathcal{G}\langle V_1,..,V_z\rangle = G\langle V_1,..,V_z\rangle \cup F\langle V_1,..,V_z\rangle.$
    \end{enumerate}
    \vspace{0mm}
    \\
    \hline
    \end{tabular}
    \caption{A baseline algorithm for computing a minimum $s$-$t$ cut of a mixed graph $\mathcal{G} = (G,F)$.}
    \label{alg:RSW+}
\end{table}

There are several key observations about the algorithm. First, we observe that $H$ is an $\varepsilon$-\emph{directed} cut sparsifier of the mixed graph $\mathcal{G} = (G, F)$. To see this, consider any cut $(S, T)$. Then,
\begin{align*}
    w_H(S,T) = |E_{G'}(S,T)| + w_F(S,T) & \leq (1+\varepsilon) \cdot |E_G(S,T)| + w_F(S,T) \\ & \leq (1+\varepsilon) \cdot \left(|E_G(S,T)| + w_F(S,T)\right) \\
    & = (1+\varepsilon) \cdot w_{\mathcal{G}}(S,T).
\end{align*}
A symmetric argument also gives $w_H(S,T) \geq (1 - \varepsilon) \cdot w_{\mathcal{G}}(S,T)$. Thus, $H$ is indeed an $\varepsilon$-directed cut sparsifier of $\mathcal{G}$. As a direct corollary, we have that $\nu(H)$ approximates $\nu(\mathcal{G})$ within a $(1 \pm \varepsilon)$ factor, following the same reasoning as in \Cref{cor:approx_nu_cq}.

Next, we observe that $K$ consists solely of undirected edges. Hence, the partition computed in Step~6 can be obtained via a recursive procedure: as long as the graph contains a connected component with edge-connectivity at least $3\varepsilon \cdot \nu(H)$, identify such a component, contract it into a single super-vertex, and recurse on the resulting graph. Upon termination, the remaining super-vertices define the desired partition $(V_1, \ldots, V_z)$. Moreover, since the contracted graph $K\langle V_1,\ldots,V_z\rangle$ is undirected and has edge-connectivity at most $3\varepsilon \cdot \nu(H)$, we can invoke \Cref{lem:small_connectivity_edge_bounds} to conclude that $|K\langle V_1, \ldots, V_z\rangle| \leq 3\varepsilon n \cdot \nu(H).$ 

We now proceed to show that the baseline algorithm correctly outputs a minimum $s$-$t$ cut of the mixed graph $\mathcal{G} = (G, F)$.

\begin{claim}  Algorithm \ref{alg:RSW+} correctly computes a minimum $s$-$t$ cut of $\mathcal{G} = (G, F)$.
\end{claim}

\begin{proof} Let $(S,T)$ be a minimum $s$-$t$ cut of $\mathcal{G}$ with value $\nu(\mathcal{G})$. It suffices to show that $(S,T)$ survives the contraction of $V_1,...,V_z$. Assume, for contradiction, that it did not. This means there exists $V_i$ such that $V_i \cap S \ne \emptyset$ and $V_i \cap T \ne \emptyset$.

Since $H$ is a directed cut sparsifier of $\mathcal{G}$, we have $\nu(\mathcal{F}) = \nu(H) \geq (1-\eps) \cdot \nu(\mathcal{G})$. Therefore $w_\mathcal{F}(S,T) \geq \nu(\mathcal{F}) \geq (1-\eps) \cdot \nu(\mathcal{G})$. Since $(S,T)$ cuts through $V_i$, we have 
$$w_K(S,T) \geq w_{K[V_i]}(S,T) \geq 3\eps \cdot \nu(H) \geq 3\eps(1-\eps) \cdot \nu(\mathcal{G})$$ where the second inequality uses the fact that $K[V_i]$ has connectivity (i.e. global min-cut value) at least $3\eps \cdot \nu(H)$. Furthermore, $K$ and $\mathcal{F}$ are disjoint and $K \cup \mathcal{F} \preceq H$. We then have:
\begin{align*}
    w_H(S,T) \geq  w_\mathcal{F}(S,T) +  w_K(S,T)
& \geq  (1-\eps) \cdot \nu(\mathcal{G}) + 3\eps(1-\eps) \cdot \nu(\mathcal{G})  \\
& = (1-\eps)(1+3\eps) \cdot \nu(\mathcal{G}).
\end{align*}
Finally, since $\mathcal{G}$ and $H$ are $(1\pm \eps)$-close in directed cut size, we then have
$$w_{\mathcal{G}}(S,T) \geq (1-\eps) \cdot w_H(S,T) \geq  (1-\eps)^2(1+3\eps) \cdot \nu(\mathcal{G}) > \nu(\mathcal{G})$$
where the final inequality follows $\eps \in (0,\frac{1}{5})$. This gives rise to a contradiction as $(S,T)$ is a minimum $s$-$t$ cut of $\mathcal{G}$ with value $w_{\mathcal{G}}(S,T) = \nu(\mathcal{G})$.
\end{proof}

\subsection{Implementation in Cut-Query}

We describe the implementation of Algorithm \ref{alg:meta-alg} in the cut-query model. As a reminder, in this setting,  we do not have direct access to the full input graph $G$; instead, we are allowed to make \emph{cut queries} to $G$. The complexity of our algorithm is measured by the total number of such queries made during its execution.

\RSWcutquery*

\begin{proof} We specify such implementation as follows.

\begin{itemize}
    \item In Step 2, we use \Cref{lem:sparsifier_cq} to w.h.p. construct $G'$ which is an $\eps$-cut sparsifier of $G$ with $\tO(n/\eps^2)$ edges and edge weights bounded by $\tO(\eps^2 n)$. 
    
    \item Once obtained $G'$, we can do Step 3,4,5,6 locally.
    \item Step 7 can be done via \Cref{lem:learn_graph_cq}.
    \item Once Step 7 is completed, we can do Step 8 locally.
\end{itemize}

It remains to bound the query complexity. Following our implementation, Step 2 consumes $\tO(n/\eps^2)$ cut queries. The number of queries in Step 7 is, up to a polylogarithmic factor, bounded by $\left|G\langle V_1,..,V_z\rangle\right|$ via \Cref{lem:learn_graph_cq}. This is also upper bounded by $(1+\eps) \cdot |G'\langle V_1,..,V_z\rangle|$ as $G'$ is a $\eps$-cut sparsifier of $G$. Since $G' \preceq K \cup \mathcal{F}$, we then have
\begin{align*}
    |G'\langle V_1,..,V_z\rangle| & \leq |K\langle V_1,..,V_z\rangle| + |\mathcal{F}| \\
    & \leq 3\eps n \cdot \nu(H) + n \cdot \sqrt{\nu(H) \cdot \tO(\eps^2 n)} \tag{\Cref{lem:small_connectivity_edge_bounds} and \Cref{lem:flow_cover}}\\
    & \leq \tO\left( \eps n^{3/2} \nu(G)^{1/2}\right) \tag{$\nu(H) \leq (1+\eps) \nu(\mathcal{G})$}
\end{align*}

Therefore, the total communication is bounded by $\tO\left(\frac{n}{\eps^2} + \eps n^{3/2} \nu(G)^{1/2}\right)$. By setting $\eps = \left(n\cdot \nu(\mathcal{G})\right)^{-1/6}$, the number of cut queries become $\tO\left(n^{4/3} \nu(\mathcal{G})^{1/3}\right)$ as wished.

As a final piece, observe that we cannot set $\eps = \left(n\cdot \nu(\mathcal{G})\right)^{-1/6}$ exactly, as we do not know the true value of $\nu(\mathcal{G})$ at the onset of our algorithm. To fix this, we can w.h.p. compute a constant-approximation of $\nu(\mathcal{G})$ within $\tO(n)$ cut queries\footnote{To do this, we first compute an $G^*$ to be a $\frac{1}{100}$-cut sparsifier of $G$ via \Cref{cor:approx_nu_cq}, and use $\nu(G^* \cup F)$ as a proxy to $\nu(\mathcal{G})$.}, and use it as a proxy to $\nu(\mathcal{G})$. This does not affect the integrity of our proofs, nor the asymptotic number of queries.
\end{proof}

\subsection{Implementation in Two-Player Communication}

We describe the implementation of Algorithm \ref{alg:meta-alg} in the two-player communication model. As a reminder, in this setting,  Alice holds a subgraph $G_A = (V, E_A)$ and Bob holds $G_B = (V, E_B)$, where the true input graph is the union $G = G_A \cup G_B$. The complexity of the algorithm is measured by the total number of bits communicated between the two players throughout the execution.

We begin by studying a useful combinatorial structure: a \emph{forest packing}. Let $G$ be an undirected, unweighted graph. For any integer $k \geq 1$, we define a \emph{$k$-forest packing} of $G$ to be $P = F_1 \cup \ldots \cup F_k$ where, for each $i \in [k]$, the subgraph $F_i$ is a spanning forest of the graph $G \setminus (F_1 \cup \ldots \cup F_{i-1})$. A crucial insight is that any $k$-forest packing yields a relatively sparse subgraph of $G$ that simultaneously preserves all small cuts and ensures that large cuts remain large. These properties are formally captured in the following lemma.

\begin{lemma} Let $P$ be $k$-forest packing of a graph $G$. Then, the following statements are true.
\begin{enumerate}
    \item For a cut $(S,T)$, if $|E_G(S,T)| \leq k$, then $E_P(S,T) = E_G(S,T)$. 
    \item For a cut $(S,T)$, if $|E_G(S,T)| \geq k+1$, then $|E_P(S,T)| \geq k$. 
    \item $P$ has average degree $O(k)$.
\end{enumerate}
\label{lem:k-cut-preserved}
\end{lemma}

\begin{proof} We first prove (1). Let $P = F_1 \cup \ldots \cup F_k$ be a forest packing of $G$. Let $(S,T)$ be a cut such that $|E_G(S,T)| \leq k$. Consider any $i \leq k$. If there is an edge among $E_G(S,T)$ remained in $G \setminus (F_1 \cup \ldots \cup F_{i-1})$, the spanning forest $F_i$ must include at least one of those remaining edges. In other words, upon packing $k \geq |E_G(S,T)|$ spanning forest into $P$, all edges of $E_G(S,T)$ must be included. This means $E_P(S,T) = E_G(S,T)$.

For (2), we have two cases. First, if $E_{F_i}(S,T) = \emptyset$ for some $i \in [k]$, it means all edges in $E_G(S,T)$ must have already been included in $F_1 \cup \ldots \cup F_{i-1}$. This then implies $|E_P(S,T)| = |E_G(S,T)| \geq k+1$. The other case is $|E_{F_i}(S,T)| \geq 1$ for each $i \in [k]$. In this case, we can write $|E_P(S,T)| = \sum_{i \in [k]} |E_{F_i}(S,T)| \geq k$ as wished.

Finally, we prove (3). Notice that each $F_i$ has at most $n-1$ edges. Thus, the number of edges in $P$ is at most $k(n-1)$, implying that the average degree is at most $\frac{2k(n-1)}{n} = O(k)$.
\end{proof}

As an implication, we obtain the following lemma.

\begin{restatable}{lem}{packingpreservesmc} Let $\mathcal{G} = (G_A \cup G_B, F)$ be a mixed graph in two-player communication setting, and let $k$ be a parameter such that $k \geq \nu(\mathcal{G})+1$. Let $P_A$ and $P_B$ be a $k$-forest packing of $G_A$ and $G_B$ respectively, and denote a mixed graph $\mathcal{P} = (P_A \cup P_B, F)$. Then, we have:
\begin{enumerate}
    \item Both $P_A$ and $P_B$ have average degree $O(k)$.
    \item $\nu(\mathcal{G}) = \nu(\mathcal{P})$.
    \item $(S,T)$ is a minimum $s$-$t$ cut of $\mathcal{G}$ if and only if it is a minimum $s$-$t$ cut of $\mathcal{P}$.  
\end{enumerate}
\label{lem:packing_preserve_mc}
\end{restatable}

\begin{proof} The average degree (1) is clear due to \Cref{lem:k-cut-preserved}. We will now show (2). Observe that $\mathcal{P} \preceq \mathcal{G}$ immediately gives $\nu(\mathcal{P}) \leq \nu(\mathcal{G})$. Let $(S,T)$ be an arbitrary minimum $s$-$t$ cut of $\mathcal{P}$ with value 
\begin{equation}
|E_{P_A}(S,T)| + |E_{P_B}(S,T)| + w_F(S,T) = w_{\mathcal{P}}(S,T)  = \nu(\mathcal{P}) \leq \nu(\mathcal{G}) \leq k-1
\label{eq2}
\end{equation}
We claim that $|E_{G_A}(S,T)| \leq k$. To see this, assume for CTD that $|E_{G_A}(S,T)| \geq k+1$. By \Cref{lem:k-cut-preserved}, we then have $|E_{P_A}(S,T)| \geq k$ which immediately contradicts (\ref{eq2}). Thus, we must have $|E_{G_A}(S,T)| \leq k$ and analogously $|E_{G_B}(S,T)| \leq k$. Then, by \Cref{lem:k-cut-preserved}, we have $E_{G_A}(S,T) = E_{P_A}(S,T)$ and $E_{G_B}(S,T) = E_{P_B}(S,T)$. This then implies
\begin{align*}
    \nu(\mathcal{G}) \leq w_\mathcal{G}(S,T) & = |E_{G_A}(S,T)| + |E_{G_B}(S,T)| + w_F(S,T) \\
    & = |E_{P_A}(S,T)| + |E_{P_B}(S,T)| + w_F(S,T) \\
    & = w_{\mathcal{P}}(S,T) \\
    & = \nu(\mathcal{P}) \tag{$(S,T)$ be a min cut of $\mathcal{P}$.}
\end{align*}
Combining this with an earlier inequality $\nu(\mathcal{P}) \leq \nu(\mathcal{G})$, we have $\nu(\mathcal{G}) = \nu(\mathcal{P})$.

Last but not lease, we prove (3). We will first prove $(\Leftarrow)$  by showing that no cuts $(S,T)$ of $\mathcal{G}$ with value $w_{\mathcal{G}}(S,T) \geq \nu(G) + 1$ becomes a min-cut of $\mathcal{P}$. Assume for contradiction that it occurred. Then, we have
\begin{align*}
    \nu(\mathcal{G}) + 1 \leq w_{\mathcal{G}}(S,T) & = |E_{G_A}(S,T)| + |E_{G_B}(S,T)| + w_F(S,T)
\end{align*}
and 
\begin{equation}
    \nu(\mathcal{G}) = \nu(\mathcal{P}) = |E_{P_A}(S,T)| + |E_{P_B}(S,T)| + w_F(S,T)   
\label{eq1}
\end{equation}
This implies that either $|E_{P_A}(S,T)| < |E_{G_A}(S,T)|$ or $|E_{P_B}(S,T)| < |E_{G_B}(S,T)|$. Assume WLOG that it is the earlier case. Then, if we had $|E_{G_A}(S,T)| \leq k$, we would have $E_{G_A}(S,T) = E_{P_A}(S,T)$ via \Cref{lem:k-cut-preserved} which is a contradiction. Thus, we must have $|E_{G_A}(S,T)| \geq k+1$. Via \Cref{lem:k-cut-preserved}, we then have $|E_{P_A}(S,T)| \geq k$. Continuing (\ref{eq1}), we then have $\nu(\mathcal{G}) \geq |E_{P_A}(S,T)| \geq k.$ This contradicts the premise that $k \geq \nu(\mathcal{G})+1$. Hence, the $(\Leftarrow)$ direction is proved.

Finally, we prove the $(\Rightarrow)$ direction. Let $(S,T)$ be an arbitrary minimum $s$-$t$ cut of $\mathcal{G}$. Then, we have
$$|E_{G_A}(S,T)| + |E_{G_B}(S,T)| + w_F(S,T) = w_{\mathcal{G}}(S,T) = \nu(\mathcal{G}) \leq k-1.$$
This implies $|E_{G_A}(S,T)| \leq k-1$. Then, via \Cref{lem:k-cut-preserved}, we have $E_{P_A}(S,T) = E_{G_A}(S,T)$ and analogously $E_{P_B}(S,T)| = E_{G_B}(S,T)$. We then have
\begin{align*}
    \nu(\mathcal{P}) = \nu(\mathcal{G}) = w_\mathcal{G}(S,T) & = |E_{G_A}(S,T)| + |E_{G_B}(S,T)| + w_F(S,T) \\
    & = |E_{P_A}(S,T)| + |E_{P_B}(S,T)| + w_F(S,T) \\
    & = w_{\mathcal{P}}(S,T)
\end{align*}
In other words, the cut $(S,T)$ attains a minimum cut value in $\mathcal{P}$; thus is a min cut of $\mathcal{P}$.
\end{proof}

This lemma proves to be extremely useful. It implies that, in order to determine a minimum $s$-$t$ cut of the graph $\mathcal{G} = (G_A \cup G_B, F)$, it suffices to determine a minimum $s$-$t$ cut of its $k$-forest packing analogue $\mathcal{P} = (P_A \cup P_B, F)$, if we set $k$ such that $k \geq \nu(\mathcal{G}) + 1$. Furthermore, Alice can locally compute $P_A$ from her subgraph $G_A$ without communicating with Bob, and vice versa. We put these advantages into action by giving a communication-efficient protocol for \Cref{lem:RSW_comm}.

\RSWcomm*

\begin{proof} We assert the following pre-processing step. 

\begin{enumerate}
    \item[(i)] Players compute $f$ be such that of $2\nu(\mathcal{G}) \leq f \leq 100 \nu(\mathcal{G})$ using $\tO(n)$ communications.\footnote{Alice computes an $G^*_A$ to be a $\frac{1}{100}$-cut sparsifier of $G_A$ using \Cref{cor:approx_nu_comm}, and Bob does the same for  $G^*_B$. Players exchange these sparsifiers and set $f := 99 \cdot \nu(G^*_A \cup G^*_B \cup F)$. }
    \item[(ii)] Alice locally computes $P_A \leftarrow f$-forest packing of $G_A$, and Bob locally computes $P_B \leftarrow f$-forest packing of $G_B$,
\end{enumerate}

According to \Cref{lem:packing_preserve_mc}, both $P_A$ and $P_B$ have average degree $O(f)$. Moreover, it suffices to determine a minimum $s$-$t$ cut of $\mathcal{P} = (P_A \cup P_B, F)$. We shall do this by implementing Algorithm \ref{alg:RSW+} with respect to $\mathcal{G} \leftarrow \mathcal{P}$.\footnote{Notationally, $G_A,G_B,G,G'$ are replaced by  $P_A,P_B,P,P'$ respectively.} We specify such implementation as follows.

\begin{itemize}
    \item In Step 2, the players proceed as follows. Alice locally construct $P'_A$ which is an $\eps$-cut sparsifier with $\tO(n/\eps^2)$ edges and maximum weight $\tO(\eps^2 f)$, and then sends it to Bob. To do that, Alice can employ a brute-force search for $P'_A$, as its existence is guaranteed due to \Cref{lem:sparsifier_small_weights}. Bob also construct $P'_B$ analogously, and then sends it to Alice. The players then set $P' = P'_A \cup P'_B $ to be an $\eps$-cut sparsifier of $P$.
    
    \item Once obtained $P'$, each player can do Step 3,4,5,6 locally.
    \item Step 7 can be done by having each player sends over all edges crossing the partition $(V_1,\ldots,V_z)$.
    \item Once Step 7 is completed, each player can do Step 8 locally.
\end{itemize}

It remains to bound the amount of communication. Following our implementation, Step 2 consumes $\tO(n/\eps^2)$ bits. The communication cost of Step 7 is, up to polylogarithmic factor, bounded by $\left|P\langle V_1,..,V_z\rangle\right|$ which is also bounded by $(1+\eps) \cdot |P'\langle V_1,..,V_z\rangle|$ as $P'$ is a $\eps$-cut sparsifier of $P$. Since $P' \preceq K \cup \mathcal{F}$, we then have
\begin{align*}
    |P'\langle V_1,..,V_z\rangle| & \leq |K\langle V_1,..,V_z\rangle| + |\mathcal{F}| \\
    & \leq 3\eps n \cdot \nu(H) + n \cdot \sqrt{\nu(H) \cdot \tO(\eps^2 f)} \tag{\Cref{lem:small_connectivity_edge_bounds} and \Cref{lem:flow_cover}}\\
    & \leq \tO\left( \eps n f\right) \tag{$\nu(H) \approx \nu(\mathcal{P}) = \nu(\mathcal{G}) \leq f$}
\end{align*}

Therefore, the total communication is bounded by $\tO\left(\frac{n}{\eps^2} + \eps n f\right)$. By setting $\eps = f^{-1/3} = \Theta\left(\nu(\mathcal{G})^{-1/3}\right)$, the communication cost becomes $\tO\left(n \cdot \nu(\mathcal{G})^{2/3}\right)$ as wished.
\end{proof}


\section{Finding a Large Flow via Cut-Query}

\label{sec:large_flow_cq}

The goal of this section is to prove \Cref{lem:large_flow_qc}, thereby concluding our minimum $s$-$t$ cut algorithm in the cut-query model. We first state our main claim.

\begin{restatable}{clm}{mainclaimcutquery} Let $f,k$ be known parameters with a promise that $\nu(G) \geq f$, and also given explicitly is a subgraph $H \preceq G$ with $\nu(H) \geq f-k$. Set $k' = \frac{k}{2}$. Then, there is a deterministic algorithm which makes $\tO\left(\frac{n^2}{\sqrt{k}}\right)$ cut queries to $G$ and outputs $H^{\textnormal{after}} \preceq G$ with $\nu(H^{\textnormal{after}}) \geq f-k' = f-\frac{k}{2}$.
\label{claim:main_claim_cq}
\end{restatable}

Such task, in fact, can be accomplished trivially with $\tO\left(n \cdot (k-k')\right) = \tO(nk)$ queries via $k-k'$ iterations of the augmenting path method. Thus, our claim yields substantial savings compared to the augmenting path approach when $k = \omega(n^{2/3})$. In addition, by recursively applying the claim, we can obtain a ``promise'' version of \Cref{lem:large_flow_qc}.

\begin{corollary} Let $f$ and $\Delta \leq f$ be parameters with a promise that $\nu(G) \geq f$. Then, there is a determinsic algorithm which makes $\tO\left(\frac{n^2}{\sqrt{\Delta}}\right)$ cut queries to $G$ and outputs $H \preceq G$ with $\nu(H) \geq f-\Delta$.
\label{cor:main_lem_cut_query}
\end{corollary} 

\begin{proof} We use \Cref{claim:main_claim_cq} recursively with $(k,k') = (f, f/2),\ldots,(2\Delta, \Delta)$ with $H$ being $H^{\text{after}}$ from the previous iteration (for the first iteration, set $H = \emptyset$). Inductively, each iteration is successful to obtain $\nu(H^{\text{after}}) \geq f-k'$. Thus, upon the completion we obtain $H$ with $\nu(H) \geq f-\Delta$. In total, the number of cut-queries made to $G$ is $\tO\left(\frac{n^2}{\sqrt{\Delta}}\right)$ due to the geometric sum.
\end{proof}

Let us now to use \Cref{cor:main_lem_cut_query} to recover \Cref{lem:large_flow_qc}.

\largeflowcutquery*

\begin{proof}
 Set $\lambda = \frac{\Delta}{10n} \leq \frac{1}{10}$. Consider the following algorithm.

\begin{enumerate}
    \item Compute $f$ such that $(1-\lambda)\nu(G) \leq f \leq (1+\lambda)\nu(G)$.
    \item[2a.] If $f \leq (1-\lambda)\Delta$, output $H := \emptyset$.
    \item[2b.] If $f \geq (1-\lambda)\Delta$, apply \Cref{cor:main_lem_cut_query} with $(f, \Delta) \leftarrow \left(\frac{f}{1+\lambda}, \frac{\Delta}{10} \right)$ and output the resulting $H$.
\end{enumerate}

Assume that Step 1 computes $f$ correctly which occurs w.h.p. We will argue that the output $H$ following step 2 is correct. If $f \leq (1-\lambda)\Delta$, we then have $\nu(G) \leq \Delta$. Hence, outputting $H := \emptyset$ with $\nu(H) = 0 \geq \nu(G) - \Delta$ is correct.

Otherwise, we have $f \geq (1-\lambda)\Delta$. In this case, we see that $\frac{f}{1+\lambda} \geq \frac{1-\lambda}{1+\lambda} \cdot \Delta > \frac{\Delta}{10}$. Therefore, the premises of \Cref{cor:main_lem_cut_query} 
 in line 2b are met. As a result, the output $H$ must satisfy:

    \begin{align*}
         \nu(H) \geq \frac{f}{1+\lambda}  -\frac{\Delta}{10} & = \nu(G) - \left[\nu(G) - \left(\frac{f}{1+\lambda} - \frac{\Delta}{10}\right) \right] \\
         & \geq  \nu(G) - \left(\frac{1}{1-\lambda} - \frac{1}{1+\lambda}\right) \cdot f - \frac{\Delta}{10} \tag{$\nu(G) \leq \frac{f}{1-\lambda}$} \\
         & \geq \nu(G) - 4 \lambda f - \frac{\Delta}{10} \tag{$\lambda \leq \frac{1}{10}$} \\
         & \geq \nu(G) - \Delta \tag{$\lambda f \leq \lambda n = \frac{\Delta }{10}$.}
    \end{align*}
Thus, the algorithm computes $H$ with a guarantee that $\nu(H) \geq \nu(G) - \Delta$.

Finally, the algorithm makes $\tO\left(\frac{n}{\lambda^2}\right)$ cut queries from step 1 and $\widetilde{O}\left(\frac{n^2}{\sqrt{\Delta}}\right)$ cut queries from step 2b. In total, these simplify to $\widetilde{O}\left(\frac{n^2}{\sqrt{\Delta}} + \frac{n^3}{\Delta^2}\right)$ queries.
\end{proof}

\paragraph{Organization.} For the remainder of this section, we will prove \Cref{claim:main_claim_cq}. We assume the its setting throughout the section: let $f,k,$ and $k' = \frac{k}{2}$ be known parameters with a promise that $\nu(G) \geq f$, and also given for free is $H \preceq G$ with $\nu(H) \geq f-k$. We wish to obtain $H^{\text{after}} \preceq G$ with $\nu(H^{\text{after}}) \geq f-k' = f-\frac{k}{2}$. We first begin with a quick reminder of our witness argument in \Cref{subsec:witness_recall}. Then, \Cref{subsec:witness_cut} presents a fresh perspective of the algorithm. \Cref{subsec:subroutine} discusses the arisen subproblem and addresses how to solve it with low cut-query complexity. Finally, \Cref{subsec:pf_main_claim} puts together all the ingredients and proves \Cref{claim:main_claim_cq}.

\subsection{Quick Recollection from \Cref{sec:setup}}
\label{subsec:witness_recall}

We now recall the definitions of witnesses and formally prove its invariant given in \Cref{sec:setup}.

\witness*

\invariant*

\begin{proof}
First, suppose that $\nu(\mathcal{G}) \geq f - k'$. Assume for the sake of contradiction that a witness exists and let $(S, T)$ be a cut with respect to that witness. By definition of a witness and $\nu(\mathcal{G})$ being the value of $s$-$t$ min-cut, we have $f-k'-1 \geq |E_\mathcal{G}(S,T)| \geq \nu(G) \geq f - k'$ which is a contradiction. Thus, $\mathcal{W}(\mathcal{G}) = \emptyset$.

Conversely, suppose that $\nu(\mathcal{G}) \leq f - k' - 1$. Again, by the definition of $\nu(G)$, we know that the min $s$-$t$ cut of $\mathcal{G}$ is of value $\nu(\mathcal{G}) \leq f - k' - 1$. Let $(S, T)$ be such cut. Let $X$ be an arbitrary subset of $(S \times T) \setminus E_{\mathcal{G}}(S,T)$ of size $f-k'-1 - |E_{\mathcal{G}}(S,T)|$. This set in fact exists because $f-k'-1 \leq n-1 \leq |S \times T|$. Then by definition, $(S, T, E_{\mathcal{G}}(S,T) \cup X)$ is a witness of $\mathcal{G}$, implying $\mathcal{W}(\mathcal{G}) \neq \emptyset$.
\end{proof}

\kill*

It is also useful to keep in mind a general description of our algorithm. Let $\gamma > 0$ be a suitable parameter. At any point that we have an explicit graph $H \preceq G$, we wish to find a $\gamma$-good edge $e \in G \setminus H$ such that  $|\mathcal{W}(H \cup \{e\})| \leq (1-\gamma) \cdot |\mathcal{W}(H)|$. Then, we set $H \leftarrow H \cup \{e\}$ and recurse until we reach the termination condition $\mathcal{W}(H) = \emptyset$ which, by our invariant, we firmly conclude $\nu(H) \geq f-k'$.

\subsection{Witness as a Cut Indicator}
\label{subsec:witness_cut}

Here, we tweak the definition of witness of $H$ by a bit: for a witness $\mathsf{W} = (S,T,Y = E_H(S,T) \cup X)$, we instead represent a cut $(S,T)$ by a binary vector $Z \in \{0,1\}^V$ which $z_v = \mathbbm{1}(v \in T)$ indicates whether $v$ belong to $t$-side of the cut. From time to time, we vary the cut representation by $(S,T)$ and $Z$ depending on the context. We always have $z_s = 0$ and  $z_t = 1$ as $s$ always appears in $S$-side of the cut, and $t$ always appears in $T$-side of the cut. With this notion, the cut-size condition of witness is equivalent to:
\begin{equation}
\sum_{(u,v) \in \mathcal{G}} |z_u - z_v| \leq f-k'-1
\label{ineq:witness_cutsize}
\end{equation}

Let us now try to reproduce \Cref{clm:good_edge_exists} in a language of $Z$. Denote $(\mathcal{Z}, \mathcal{Y})$ to be a uniform distribution over all witnesses of $H$. To draw a uniformly random witness $(Z,Y) \sim (\mathcal{Z}, \mathcal{Y})$, we can (1) draw $Z \sim \mathcal{Z}$, and (2) draw $Y \sim \mathcal{Y}\mid Z$. We also define a \emph{real-valued} vector $\hat{Z} := \{\hat{z}_v\}_{v \in V} \in [0,1]^V$ such that for any vertex $v \in V$, we let:
$$\hat{z}_v := \mathop{\E}_{Z \sim \mathcal{Z}} z_v.$$ Since $z_s$ is always $0$ and $z_t$ is always $1$, it follows that $\hat{z}_s = 0$ and $\hat{z}_t = 1$. Intuitively, the vector $\hat{Z}$ can be viewed as an \emph{average cut} over all remaining witnesses in $\mathcal{W}(H)$. In particular, $\hat{z}_v$ admits a natural probabilistic interpretation: it corresponds to the probability that $v \in T$ under a uniformly random witness $(S, T, Y)$. Importantly, we remark that the explicit description of the $H$ alone is sufficient to enumerate all remaining witnesses in $\mathcal{W}(H)$, and thus to compute $\hat{Z}$.

\begin{claim} For any pair of vertices $(u,v) \notin H$, we have $ \mathop{\Pr}_{Z \sim \mathcal{Z}}(z_u \ne z_v) \geq |\hat{z}_u - \hat{z}_v|$.
\label{clm:diff_side}
\end{claim}

\begin{proof} Such probability is  $$ \mathop{\E}_{Z \sim \mathcal{Z}} |z_u - z_v| \geq \left|  \mathop{\E}_{Z \sim \mathcal{Z}} z_u -   \mathop{\E}_{Z \sim \mathcal{Z}} z_v\right| = |\hat{z}_u - \hat{z}_v|$$ via the triangle inequality.
\end{proof}

\begin{claim} Fix a pair of vertices $(u,v) \notin H$. Let $Z \in supp(\mathcal{Z} \mid z_u \ne z_v)$. Then,
$$ \mathop{\Pr}_{Y \sim \mathcal{Y} \mid Z}\left[(u,v) \notin Y \setminus E_H(S,T) \mid Z\right] \geq \frac{k'+1}{k}.$$
\label{clm:kill_condition_Z}
\end{claim}

\begin{proof}  Observe that $\mathcal{Y}\mid Z$ is equivalent to $E_H(S,T) \cup (\mathcal{X} \mid Z)$ where $\mathcal{X} \mid Z$ is a uniform distribution over $\binom{(S \times T) \setminus E_H(S,T)}{f-k'-1-E_H(S,T)}$. Note that $(u,v)$ is one of those pair of vertices among $(S \times T) \setminus E_H(S,T)$. Thus, the probability that $(u,v)$ is included in $X$ is 

\begin{align*}
    \frac{f-k'-1 - |E_H(S,T)|}{|S \times T| - |E_H(S,T)|} & \leq \frac{f-k'-1 - |E_H(S,T)|}{n-1 - |E_H(S,T)|} \tag{$|S \times T| \geq n-1$} \\
    & \leq \frac{f-k'-1 - (f-k)}{n-1 - (f-k)} \tag{$|E_H(S,T)| \geq \nu(H) \geq f-k$} \\
    & \leq \frac{k-k'-1}{k} \tag{$f \leq n-1$}
\end{align*}
Subtracting both sides from 1 concludes the proof.   
\end{proof}

These two claims have important implications. The left-hand side of \Cref{clm:diff_side} can be interpreted as the probability that $u$ and $v$ lie on opposite sides of a uniformly random witness cut $Z$. Additionally, the left-hand side of \Cref{clm:kill_condition_Z} can be viewed as the probability that, conditioned on a fixed choice of $Z$ with $u$ and $v$ on the different sides, a random witness is \emph{killed} by the edge $(u,v)$. By combining these two interpretations, we arrive the following bound.

\begin{claim} For any edge $(u,v) \in G \setminus H$, we have $$ \mathop{\Pr}_{(Z,Y) \sim (\mathcal{Z},\mathcal{Y})}\left[(u,v) \text{ kills } (Z,Y) \right] \geq \frac{k'+1}{k} \cdot \left|\hat{z}_u - \hat{z}_v\right|.$$
\label{clm:frac_witness_killed_cq}
\end{claim}

\begin{proof} Recall that, by definition, an edge $(u, v)$ \emph{kills} a witness $(S, T, Y)$ if $z_u \neq z_v$—i.e., $u$ and $v$ lie on opposite sides of the cut—and $(u, v) \notin Y$. Under this characterization, the LHS quantity becomes:
\begin{align*}
    \mathop{\Pr}_{Z \sim \mathcal{Z}}(z_u \ne z_v) \cdot  \mathop{\Pr}_{(Z,Y) \sim(\mathcal{Z},\mathcal{Y} \mid z_u \ne z_v)} \left[(u,v) \text{ kills } (Z,Y)  \mid z_u \ne z_v\right].
\end{align*}
By \Cref{clm:diff_side}, we have $ \mathop{\Pr}_{Z \sim \mathcal{Z}}(z_u \ne z_v) 
 \geq \left|\hat{z}_u - \hat{z}_v\right|$. Moreover, we can bound the second term by:
\begin{align*}
    &  \mathop{\Pr}_{(Z,Y) \sim(\mathcal{Z},\mathcal{Y} \mid z_u \ne z_v)} \left[(u,v) \text{ kills } (Z,Y)  \mid z_u \ne z_v\right] \\
    & = \sum_{Z \in \text{supp}(\mathcal{Z} \mid z_u \ne z_v)} \Pr \left(\mathcal{Z} = Z \mid z_u \ne z_v\right) \cdot  \mathop{\Pr}_{Y \sim \mathcal{Y} \mid Z} \left[(u,v) \notin Y \setminus E_H(S, T) \mid Z\right]) \\
    & \geq \sum_{Z \in \text{supp}(\mathcal{Z} \mid z_u \ne z_v)} \Pr \left(\mathcal{Z} = Z \mid z_u \ne z_v\right) \cdot \frac{k'+1}{k} \tag{\Cref{clm:kill_condition_Z}} \\
    & = \frac{k'+1}{k}  \cdot \sum_{Z \in \text{supp}(\mathcal{Z} \mid z_u \ne z_v)} \Pr \left(\mathcal{Z} = Z \mid z_u \ne z_v\right) \\
    & = \frac{k'+1}{k}.
\end{align*}
Multiplying the two bounds completes the proof.
\end{proof}

We emphasize that \Cref{clm:frac_witness_killed_cq} is vastly useful in our analysis, as it applies to every edge $(u,v) \in G \setminus H$. Recall that the distribution $(\mathcal{Z}, \mathcal{Y})$ is a uniform among all remaining witnesses. Hence, the claim then implies that any edge $(u,v) \in G \setminus H$ kills an $\Omega\left(|\hat{z}_u - \hat{z}_v|\right)$ fraction of the current witnesses. Furthermore, given the current subgraph $H$, the values $\hat{z}_v$ for all $v \in V$ can be computed locally with no extra costs. As a result, our task is significantly simplified: it suffices to identify an edge $(u,v) \in G \setminus H$ such that $|\hat{z}_u - \hat{z}_v| \geq \gamma$ for some suitable threshold $\gamma > 0$. The following claim ensures the existence of such an edge, with $\gamma = \Omega\left(\frac{\sqrt{k}}{n}\right)$.

\begin{claim} There exists an edge $(u,v) \in G \setminus H$ such that $|\hat{z}_u - \hat{z}_v| \geq \alpha \cdot \frac{k'}{n\sqrt{k}}$, for some universal constant $\alpha > 0$.  
\label{claim:far_edge_exists} 
\end{claim}

\begin{proof} For any support $Z \in \text{supp}(\mathcal{Z})$ (meaning that $Z = (S,T)$ appears in some witness), we have $\sum_{(u,v) \in H} |z_u - z_v| \leq f-k'-1$ following the equivalent definition of witness in (\ref{ineq:witness_cutsize}). Taking expectation over $Z \sim \mathcal{Z}$, we have:
\begin{align*} f-k'-1 & \geq  \mathop{\E}_{Z \sim \mathcal{Z}} \sum_{(u,v) \in H} |z_u - z_v| \\
& = \sum_{(u,v) \in H}  \mathop{\E}_{Z \sim \mathcal{Z}}  |z_u - z_v| \tag{linearlity of expectation} \\
& \geq \sum_{(u,v) \in H} \left|  \mathop{\E}_{Z \sim \mathcal{Z}} z_u -   \mathop{\E}_{Z \sim \mathcal{Z}} z_v\right| \tag{triangle inequality} \\
& = \sum_{(u,v) \in H} |\hat{z}_u - \hat{z}_v|.
\end{align*}

On the other hand, let $Q \subseteq G \setminus H$ be such that $|Q| = O(n\sqrt{k})$ and $\nu(H \cup Q) = f$. The existence of $Q$ is guaranteed via construction. Let $F \preceq H$ be a flow subgraph of $H$ with $\nu(F) = \nu(H) \geq f-k$. We then have $\nu(G_F) = \nu(G) - \nu(F) \geq f - \nu(F)$. We then let $F'$ be a non-circular flow of $G_F$ with value $f - \nu(F)$ and set $Q = F' \setminus H$. We then see that $\nu(H \cup Q) \geq \nu(F') + \nu(F) = f$ and $|Q| \leq |F'| \leq O(n\sqrt{f - \nu(F)}\ ) \leq O(n\sqrt{k})$ due to \Cref{lem:flow_cover}. As a result, we derive:
\begin{align*}
    \sum_{(u,v) \in H \cup Q} |\hat{z}_u - \hat{z}_v| & \geq \sum_{i \in [f]} \sum_{\text{edges $(u,v)$ in $i^{\text{th}}$ $s$-$t$ flow}} |\hat{z}_u - \hat{z}_v| \\
    &  \geq \sum_{i \in [f]} |\hat{z}_t - \hat{z}_s| \tag{triangle inequality} \\
    & = f \tag{$\hat{z}_s = 0$ and $\hat{z}_t = 1$.}
\end{align*}
This implies
$$\sum_{(u,v) \in Q} |\hat{z}_u - \hat{z}_v| \geq k'+1$$
which means some edge $(u,v) \in Q \subseteq G \setminus H$ must have $|\hat{z}_u - \hat{z}_v| \geq \frac{k'+1}{|Q|} \geq  \Omega\left(\frac{k'}{n\sqrt{k}}\right)$.
\end{proof}

By combining \Cref{clm:frac_witness_killed_cq} and \Cref{claim:far_edge_exists}, there must exist of an edge $(u,v) \in G \setminus H$ that kills a $\frac{k'+1}{k} \cdot \left(\alpha \cdot \frac{k'}{n\sqrt{k}}\right) \approx \frac{\alpha}{4} \cdot \frac{\sqrt{k}}{n}$ fraction of the remaining witnesses, via the fact that $k' = \frac{k}{2}$. In the next section, we introduce a subroutine that, using only $O(\log n)$ cut queries, enables us to find an edge in $G \setminus H$ that is (asymptotically) as effective as such edge $(u,v)$.

\subsection{Useful Subroutines}
\label{subsec:subroutine}
\newcommand{\FLE}{\mathsf{FindLongEdge}}

Consider the following problem called $\FLE$. We are given the following items as an input. 
\begin{itemize}
    \item An underlying graph $G = (V,E)$ to which we are able to make cut-query.
    \item A potential function $\phi: V \rightarrow [0,1]$ that maps $v \in V$ to its position between $0$ and $1$.
    \item A threshold $\delta \in (0,1)$.
\end{itemize}

There is also a hidden parameter $\delta_{\text{max}} := \max_{(u,v) \in E} |\phi(u) - \phi(v)|$ not known to us. We are asked to identify a \emph{long} edge—if one exists—whose endpoints are mapped far apart, by making cut queries to $G$.  Formally, we must output either
\begin{itemize}
    \item $(u,v) \in G$ with $|\phi(u) - \phi(v)| \geq \frac{\delta}{2}$, or
    \item ``FAIL'' which means $\delta > \delta_{\text{max}}$. 
\end{itemize}
Another helpful interpretation is as follows. If $\delta \leq \delta_{\text{max}}$, then we must output $(u,v) \in G$ with $|\phi(u) - \phi(v)| \geq \delta/2$. On the other hand, if $\delta > \delta_{\text{max}}$, it is acceptable to output either such $(u,v)$, if existed, or ``FAIL''.

We first describe a useful object that will shortly help us design an algorithm for $\FLE$. Let $\ell$ be an arbitrary positive integer. We call $A = \left((p_1,q_1),...,(p_\ell, q_\ell)\right)$ an \emph{assignment} iff for each $i$ we have $(p_i, q_i) \in \{1,2,3\} \times \{1,2\}$ and $p_i \ne p_{i+1}$. We say a pair $(i, j) \in [\ell]^2$ is \emph{matched} by $A$ iff $p_i = p_j$ and $q_i \ne q_j$. By the description of assignments, we notice that any pairs $(i,j)$ with $|i-j| \leq 1$ is never matched. The following claim shows that all remaining pairs can be matched at least once if we are allowed to have $O(\log \ell)$ assignments.

\begin{claim} There exists a set of $t = O(\log \ell)$ assignments $\mathcal{A} = \{A^{(1)},...,A^{(t)}\}$ such that every pair $(i,j) \in [\ell]^2$ with $|i-j| \geq 2$ is matched by at least one assignment among $\mathcal{A}.$ 

    \label{clm:far_matched}
\end{claim}
\begin{proof} We will prove the existence via a probabilistic argument. Consider the following random process for constructing an assignment $A = ((p_1,q_1),...,(p_\ell,q_\ell))$.
\begin{enumerate}
    \item Assign $p_1 = 1$. 
    \item For any $2 \leq i \leq \ell$, sequentially assign $p_i$ uniformly at random among $\{1,2,3\} \setminus \{p_{i-1}\}$
    \item For each $j$, assign $q_j \in \{1,2\}$ independently.
\end{enumerate}

We make the following miniclaim: for any $(i,j)$ that $|i-j| \geq 2$, we have $(i,j)$ being matched by $A$ with probability at least $\frac{1}{8}$. To see this, assume WLOG that $i \geq j+2$. Then, we have
\begin{align*}
    \Pr(p_i = p_j) & = \Pr(p_{i-1} = p_j) \cdot \Pr(p_i = p_j \mid p_{i-1} = p_j) +  \Pr(p_{i-1} \ne p_j) \cdot \Pr(p_i = p_j \mid p_{i-1} \ne p_j) \\
    & = \frac{1}{2} \cdot \Pr(p_{i-1} \ne p_j)
\end{align*}
since $ \Pr(p_i = p_j \mid p_{i-1} = p_j)  = 0$ and $\Pr(p_i = p_j \mid p_{i-1} \ne p_j) = \frac{1}{2}$. Moreover, we have
\begin{align*}
    & \Pr(p_{i-1} \ne p_j) \\
    & = \Pr(p_{i-2} = p_j) \cdot \Pr(p_{i-1} \ne p_j \mid p_{i-2} = p_j) +  \Pr(p_{i-2} \ne p_j) \cdot \Pr(p_{i-1} \ne p_j\mid p_{i-1} \ne p_j) \\
    & = \Pr(p_{i-2} = p_j) + \frac{1}{2} \cdot \Pr(p_{i-2} \ne p_j) \\
    & = \frac{1}{2} + \frac{1}{2}\cdot \Pr(p_{i-2} = p_j).
\end{align*}
where the second equality is true via the construction: we have $\Pr(p_{i-1} \ne p_j \mid p_{i-2} = p_j) = 1$ and $\Pr(p_{i-1} \ne p_j\mid p_{i-1} \ne p_j) = \frac{1}{2}$. It is worth noting that these chain of equations subtly relies on the assumption that $i \geq j + 2$. This condition ensures that our random process assigns a value to $p_j, \ldots, p_{i-2},p_{i-1},$ and $p_i$ in that particular order.

Combining the two equations, we have $\Pr(p_i = p_j) =\frac{1}{4} + \frac{1}{4} \cdot \Pr(p_{i-2} = p_j) \geq \frac{1}{4}$. Since $q_i$ and $q_j$ are assigned uniformly and independently among $\{1,2\}$, and independent of $p_i$ and $p_j$, we have
$$\Pr(p_i = p_j \land q_i \ne q_j) = \Pr(p_i = p_j) \cdot \Pr(q_i \ne q_j) \geq \frac{1}{8}.$$ In other words, $(i,j)$ are matched by $A$ with probability at least $\frac{1}{8}$, thereby proving our miniclaim.

Now consider drawing $t = 24 \log n$ independently random assignments from the above process. For any pair $(i,j) \in [\ell]^2$ with $|i-j| \geq 2$, the probability that $(i,j)$ is not matched by any of the $t$ assignments is at most $$\left(\frac{7}{8}\right)^{24 \log n} \leq n^{-3}.$$ Applying the union bound over all $O(n^2)$ possible pairs of $(i,j)$'s, w.h.p. every such pairs is matched by at least one among the $t$ assignments. As a consequence, there must be a deterministic choice of those $t$ assignments that the statement remains true, thereby proving our claim.
\end{proof}

We are now ready to give an efficient cut-query algorithm for $\FLE$.

\begin{lemma}
    $\FLE$ can be solved deterministically within $O\left(\log{(n/\delta)}\right)$ cut queries to $G$. 
\label{lem:subroutine}
\end{lemma}

\begin{proof} Partition $[0,1]$ into $\ell = \lceil 2/\delta\rceil $ intervals of length $\delta/2$ (besides the last interval which may be shorter). Denote the $j^{\text{th}}$ interval by $I_j := \left[\frac{(j-1)\delta}{2}, \frac{j\delta}{2}\right)$. Let $\mathcal{A}$ be a set of $O(\log \ell)$ assignments that follows \Cref{clm:far_matched} with respect to $\ell$. Our algorithm proceeds as follows.

\begin{enumerate}
    \item[(1)] For each assignment $A = ((p_1,q_1),...,(p_\ell, q_\ell)) \in \mathcal{A}$, we do:
    \begin{enumerate}
        \item [(1a)] Partition $V$ into 6 sets. Namely, for each $(r,s) \in \{1,2,3\} \times \{1,2\}$, we denote:
    $$V_{r,s} := \{v \in V \mid \phi(v) \in I_j \land (r,s) = (p_j,q_j)\}.$$
    In words, for each $v \in V$, we first identify the block $I_j$ that contains $\phi(v)$. Then we look up the $j^{th}$ pair of the assignment, namely $(p_j,q_j)$, and that becomes the index of the partition which $v$ belongs to.
    \item [(1b)] Make three find-edge queries: $\FindEdge(V_{1,1}, V_{1,2}), \FindEdge(V_{2,1}, V_{2,2}), $ and $\FindEdge(V_{3,1}, V_{3,2})$.
    \item [(1c)] If any of the three find-edge queries ever return an edge $(u,v)$, output it and terminate immediately. Otherwise, proceed to the next assignment.
    \end{enumerate}
    \item[(2)] If all find-edge queries never return an edge, output ``FAIL''.
    
\end{enumerate}
We now justify the correctness of our algorithm, which follows from these two miniclaims.
\begin{enumerate}

    \item If $\delta \leq \delta_{\text{max}}$, then the algorithm never outputs ``FAIL'', i.e. it must output some edge $(u,v)$.
    
    Let $(u^*,v^*) \in E$ be such that $|\phi(u^*) - \phi(v^*)| = \delta_{\text{max}}$. Let $\phi(u^*) \in I_a$ and $\phi(v^*) \in I_b$ for some $a,b \in [\ell]$. WLOG, assume that $\phi(u^*) \geq  \phi(v^*)$ so that $\phi(u^*) - \phi(v^*) = \delta_{\text{max}}$. Then, we have
    \begin{align*}
        \delta \leq \delta_{\text{max}} = \phi(u^*) - \phi(v^*) < \frac{\delta}{2} \cdot (a+1) - \frac{\delta}{2} \cdot b.
    \end{align*}
    This then implies that $a-b \geq 2$. By the construction of our assignments as guaranteed by \Cref{clm:far_matched}, $(a,b)$ must be matched by some assignment $((p_1,q_1),...,(p_\ell,q_\ell)) \in \mathcal{A}$. Within this assignment, we must have $p_a = p_b := p$ and $q_a \ne q_b$. When we make the query $\FindEdge(V_{p,1}, V_{p,2})$ in (1b), the vertices $u^*$ and $v^*$ belong to the different side of the query. As a result, the find-edge query must return some edge $(u, v) \in V_{p,1} \times V_{p,2}$ which may or may not be the same edges as $(u^*,v^*)$. Thus, the algorithm never outputs ``FAIL''.
    
    \item If the algorithm outputs some edge $(u,v)$, then we have $|\phi(u) - \phi(v)| \geq \frac{\delta}{2}$.
    
    To see this, assume that $(u,v)$ is found upon a find-edge query $\FindEdge(V_{i,1}, V_{i,2})$ of an assignment $A \in \mathcal{A}$. This implies $u \in V_{i,1}$ and $v \in V_{i,2}$. Let $a,b \in [\ell]$ such that $\phi(u) \in I_a$ and $\phi(v) \in I_b$. By definition, we have $(p_a,q_a) = (i,1)$ and $(p_b,q_b) = (i,2)$ which implies $a \ne b$. Moreover, since adjacent interval cannot have the same $p$ values, we must have $|a-b| \geq 2$. WLOG we assume that $a-b \geq 2$. As a result we have 
    $$\phi(u) - \phi(v) \geq \frac{\delta}{2}\cdot (a-1) - \frac{\delta}{2} \cdot b \geq \frac{\delta}{2}$$
    as wished.
\end{enumerate}

Finally, we analyze the query complexity of our algorithm. We make up to $O(\log{\ell}) = O(\log{(1/\delta)})$ find-edge queries, out of which up to one query outputs an edge. Since all find-edge queries that answer ``NONE'' spends $O(1)$ cut queries, and the one that actually outputs an edge spends $O(\log n)$ cut queries, in total our algorithm only makes $O(\log{(n/\delta)})$ cut queries to $G$.
\end{proof}

\subsection{Putting All Together: Proof of \Cref{claim:main_claim_cq}}
\label{subsec:pf_main_claim}
We now combine all the pieces to prove our main claim of this section. 

\mainclaimcutquery*

\begin{proof} We do as follows.
\begin{enumerate}
    \item If $\W(H) = \emptyset$, output the current $H$ and terminate. Otherwise, proceed. 
    \item Compute $\hat{z}_v$ for all $v \in V$ w.r.t. $H$.
    \item Use the algorithm of \Cref{lem:subroutine} to compute $\FLE$ with respect to the graph $G \setminus H$, potential $\phi(v) := \hat{z}_v$ and $\delta := \frac{\alpha}{2} \cdot \frac{\sqrt{k}}{n}$. Let $(u,v) \in G \setminus H$ be the edge that is outputted.
    \item Set $H \leftarrow H \cup (u,v)$ and go back to Step 1.
\end{enumerate}

If our algorithm ever terminates, its correctness follows from \Cref{invariant} that $\W(H) = \emptyset$ implies $\nu(H) \geq f-k'$. Thus, it suffice to justify the termination and query complexity.

Assume that at the start of an iteration, we have $\W(H) \ne \emptyset$. \Cref{claim:far_edge_exists} tells us that some $(u^*,v^*) \in G \setminus H$ has $|\hat{z}_{u^*} - \hat{z}_{v^*}| \geq \frac{\alpha}{2} \cdot \frac{\sqrt{k}}{n}$; i.e. the graph $G \setminus H$ has $\delta_\text{max} \geq \frac{\alpha}{2} \cdot \frac{\sqrt{k}}{n} = \delta$. This means Step 2 successfully finds $(u,v) \in G \setminus H$ such that $|\hat{z}_u - \hat{z}_v| \geq  \frac{\delta}{2} = \frac{\alpha}{4} \cdot \frac{\sqrt{k}}{n}$ per the guarantee of \Cref{lem:subroutine}. Thus, via \Cref{clm:frac_witness_killed_cq}, the fraction of $\W(H)$ that $(u,v)$ kills is at least $ \frac{k'+1}{k} \cdot |\hat{z}_u - \hat{z}_v| \geq \frac{\alpha}{8}\cdot \frac{\sqrt{k}}{n}$. As a result, upon inserting $(u,v)$ to $H$, the number of remaining witnesses is reduced to $$\W(H \cup (u,v)) \leq \left(1-\frac{\alpha}{8}\cdot \frac{\sqrt{k}}{n}\right) \cdot \W(H) \leq \exp\left(- \frac{\alpha}{8}\cdot \frac{\sqrt{k}}{n}\right) \cdot \W(H).$$

A trivial bound shows that in the beginning, we have at most $n^{O(n)}$ witnesses. Therefore, we only need 
$$\left(\frac{8}{\alpha} \cdot \frac{n}{\sqrt{k}} \right)\cdot \log{\left(n^{O(n)}\right)} = \tO\left(\frac{n^2}{\sqrt{k}}\right)$$ iterations before all witnesses are killed. Then, in the next iteration, the algorithm will surely terminate.

Finally, we need to bound the number of cut queries made to $G$. Consider an arbitrary iteration which $H$ has been given explicitly. Only Line 3  requires us to make $O(\log n)$ cut queries to $G \setminus H$. Each of which can be simulated by a single query to $G$, along with an explicit description of $H$. Since the algorithm only needs $\tO\left(\frac{n^2}{\sqrt{k}}\right)$ iterations, we only make $\tO\left(\frac{n^2}{\sqrt{k}}\right)$ cut queries to $G$.
\end{proof}


\section{Open Problems}
\label{sec:open_problems}

We conclude with a brief discussion of open problems and future directions motivated by our results.

\paragraph{SFM Lower Bounds via Graph Cuts.}
%
The current best deterministic lower bound for submodular function minimization (SFM) is due to Chakrabarty, Graur, Jiang, and Sidford~\cite{ChakrabartyGJS22}, who showed that any algorithm must make at least $\Omega(n \log n)$ value queries to solve SFM in the value oracle model. Their construction involves a carefully designed submodular function tailored to their lower bound framework. This raises the intriguing open question of whether any variants of graph cut functions can be used to establish a matching $\Omega(n \log n)$ lower bound. In particular, obtaining an $\Omega(n^{1.01})$ lower bound for SFM, whether via graph cuts or not, would be considered a major breakthrough.

\paragraph{SFM Upper Bound and Directed/Weighted Graphs.} 
On the upper bound side of SFM, the best-known algorithm by Jiang \cite{Jiang23} makes $\tO(n^2)$ value queries. An obvious major open problem is to achieve, e.g., $O(n^{2-\epsilon})$ bound for some constant $\epsilon>0$. Towards this goal, computing the minimum $s$--$t$ cuts on directed/weighted graphs presents a major barrier as {\em no} $o(n^2)$ algorithm was known for {\em either} the case of directed unweighted graphs and case of undirected weighted graphs (despite many  progress in the sequential and other settings \cite{BrandGJV25,Brand0KLMGS24, Brand0PKLGSS23,ChenKLPGS22,BrandGJLLPS22,BrandLLSS0W21,BrandLNPSS0W20}). We believe that presenting an $O(n^{2-\epsilon})$ bound for one of these cases will be a major progress. In particular, for the directed case, even a more apparent barrier is to solving the following simple {\em reachability} problem \cite{Nanongkai24}: Can we determine if there is a {\em directed} $s$--$t$ path in an unweighted directed graph in $O(n^{2-\epsilon})$ cut queries for some constant $\epsilon>0$? It should be a breakthrough already to solve this reachability problem.

For the {\em global} min-cut problem, the cut query complexity is essentially resolved for weighted undirected graphs~\cite{MukhopadhyayN20}. While there has been progress in the standard sequential setting~\cite{Cen0NPSQ21}, we are not aware of any new developments in the cut query model. We leave this as another important open problem.


\paragraph{Finding a Maximum Flow.}
The duality between the minimum $s$--$t$ cut and the maximum $s$--$t$ flow is one of the most celebrated results in combinatorial optimization. While our algorithm computes (and enumerates all) minimum $s$--$t$ cuts, it does not immediately yield a maximum $s$--$t$ flow. This limitation is also shared by the algorithm of Rubinstein, Schramm, and Weinberg~\cite{RubinsteinSW18}. On the other hand, the deterministic algorithm of Anand, Saranurak, and Wang~\cite{AnandSW25} does output a maximum $s$--$t$ flow using $\widetilde{O}(n^{5/3})$ cut queries and communication. Interestingly, a closer examination of our approach reveals that it too can be extended to recover a maximum flow with complexity $\tO(n^{5/3})$: by applying either \Cref{lem:large_flow_qc} or \Cref{lem:large_flow_comm} with threshold $\Delta = n^{2/3}$, and then augmenting to a full maximum flow using $O(n^{2/3})$ additional iterations of the standard augmenting path method. Nevertheless, it remains an intriguing open problem whether one can compute a maximum $s$--$t$ flow using fewer than $o(n^{5/3})$ cut queries or bits of communication. Doing so with {\em combinatorial algorithms}, even for the special case of bipartite matching \cite{BlikstadBEMN22} will also be interesting.

\paragraph{Improving Query Complexity.}
Our $\widetilde{O}(n^{8/5})$-query algorithm (\Cref{thm:min_cut_cq}) represents the first progress beyond the $\widetilde{O}(n^{5/3})$ barrier established by ~\cite{RubinsteinSW18} and ~\cite{AnandSW25}, demonstrating that this bound is not optimal. A natural next milestone is to design an algorithm with complexity $\widetilde{O}(n^{3/2})$. Although the support size of a minimum $s$--$t$ cut is at most $O(n)$, it is reasonable to conjecture a lower bound of $\widetilde{\Omega}(n^{3/2})$ for the problem, as this matches the support size of its dual—the maximum $s$--$t$ flow of value $\Omega(n)$. In fact, one can show that any algorithm which outputs an exact maximum flow—or even a multiplicative $\Omega(1)$-approximation—must incur a communication complexity of $\Omega(n^{3/2})$ bits, and therefore a cut-query complexity of $\widetilde{\Omega}(n^{3/2})$. From the lower bound perspective, showing that any deterministic two-player communication protocol for exactly computing the \emph{value} of the minimum $s$--$t$ cut requires $\widetilde{\Omega}(n^{3/2})$ bits of communication would be an exciting result.


\paragraph{Improving Combinatorial Protocols.} Similarly, it is natural to ask whether a $\widetilde{O}(n^{3/2})$-bit protocol for computing minimum $s$--$t$ cuts can be achieved using combinatorial techniques. The work of \cite{HY25} confirms that such an upper bound is indeed possible, but their protocol relies on non-combinatorial continuous optimization approaches. We leave it as an open question whether this $\tO(n^{3/2})$ bound can be matched by a fully combinatorial protocol. Our $\widetilde{O}(n^{11/7})$-bit deterministic protocol (\Cref{thm:min_cut_comm}) represents progress in this direction.

\paragraph{Deterministic Algorithms.} While our $\widetilde{O}(n^{11/7})$ protocol for the two-party communication model is fully deterministic, our $\widetilde{O}(n^{8/5})$ cut-query algorithm is randomized. To date, the only known deterministic cut-query algorithm for computing a minimum $s$--$t$ cut is the $\widetilde{O}(n^{5/3})$-query algorithm of ~\cite{AnandSW25}, which is based on a careful implementation of Dinitz's blocking flow method. It remains an open question whether this $\widetilde{O}(n^{5/3})$ bound can be improved in the deterministic setting.

Notably, our $\widetilde{O}(n^{8/5})$ algorithm is deterministic except for one component: the computation of an $\varepsilon$-cut sparsifier. Hence, a deterministic algorithm for constructing such a sparsifier using $\widetilde{O}(n / \varepsilon^2)$ cut queries would immediately yield a fully deterministic $\widetilde{O}(n^{8/5})$-query algorithm for the minimum $s$--$t$ cut problem under our framework. Moreover, via the reduction of~\cite{AnandSW25}, such an improvement would also extend to computing global minimum cuts deterministically within the same query complexity, up to polylogarithmic factors.

\paragraph{All-Pairs Min-Cut.} Another natural problem in the cut-query model is to compute the \emph{all-pairs} min-cut of a graph. Simultaneous to our work, Kenneth-Mordoch and Krauthgamer \cite{KK_personal} presented the first non-trivial algorithm for this problem, requiring $\tO(n^{7/4})$ randomized cut queries. Given the recent progress in the sequential setting, where the complexity of all-pairs min-cut has been nearly matched to that of an $s$–$t$ min-cut \cite{Abboud0PS23, AbboudKLPPSYY25, GutenbergKYY25}, it is natural to ask whether a similar equivalence can be achieved in the cut-query model. A weaker, yet still interesting, open question is whether one can improve upon the complexity of $\tO(n^{7/4})$ in either the cut-query or two-player communication setting, which is already an interesting problem in its own right.

\printbibliography

\appendix 

\section{Missing Proofs from \Cref{sec:prelim}}
\label{appendix:flow_cover_proof}

\flowcover*

\begin{proof} 
Since $F$ contains no cycles, there exists an topological ordering $\phi: V \rightarrow [n]$ such that for any directed edge $(a,b) \in E$, we have $\phi(a) < \phi(b)$. For any edge $e = (a,b) \in E$, denote $\phi(e) = \phi(b) - \phi(a)$. We shall also write $e$ and $(a,b)$ interchangeably.

We will compute $X = \sum_{e \in E} w_F(e) \cdot \phi(e)$ in two different ways:

\begin{enumerate}
    \item Let the $f$ units of flow of $F$ be $\mathcal{F}_1, \ldots, \mathcal{F}_f$ where each flow is a (disjoint) directed path from $s$ to $t$. Consider any flow $\mathcal{F}_i = (s = v_0, v_1, \ldots, v_d = t)$. Then, we have:

    $$\sum_{e \in \mathcal{F}_i} \phi(e) = \sum_{j = 0}^{d-1} \phi(v_{j+1}) - \phi(v_j) = \phi(v_d) - \phi(v_0) \leq n.$$
    Since $F = \bigcup_{i=1}^f \mathcal{F}_i$, we can compute $X$ by summing over the flows:
\begin{equation}X = \sum_{e \in E} w_F(e) \cdot \phi(e) = \sum_{i = 1}^{f} \sum_{e \in \mathcal{F}_i} \phi(e) \leq \sum_{i = 1}^{f} n = nf.
\label{eq:flowcover1}
\end{equation}

    \item For any vertex $v \in V$, let $N(v)$ denote the set of out-neighbors of $v$. Then, we can compute $X$ by summing over the vertices:
    \begin{align*}
         X = \sum_{e \in E} w_F(e) \cdot \phi(e) & = \sum_{v \in V} \sum_{u \in N(v)} w_F(v,u) \cdot \phi(v,u) \\
         &\geq \sum_{v \in V} \sum_{u \in N(v)} \phi(v,u) \tag{weights are at least 1} \\
         &= \sum_{v \in V} \sum_{u \in N(v)} \left(\phi(u) - \phi(v)\right).
    \end{align*}
    Since $\phi$ is an injection and is a topological ordering, the values $\phi(u) - \phi(v)$ among $u \in N(v)$ are all distinct positive integers. Therefore, we have:
    $$\sum_{u \in N(v)} \left(\phi(u) - \phi(v)\right) \geq \sum_{j = 1}^{|N(v)|} j \geq \frac{1}{2} \cdot |N(v)|^2.$$
    Continuing the bound for $X$, we get:
    \begin{equation}
        X \geq \sum_{v \in V} \frac{1}{2} \cdot |N(v)|^2 \geq \frac{1}{2n} \cdot |E|^2
    \label{eq:flowcover2}
    \end{equation}
    where the last inequality follows from Cauchy-Schwarz and the fact that $|E| = \sum_{v \in V} |N(v)|$.
\end{enumerate}
Combining (\ref{eq:flowcover1}) and (\ref{eq:flowcover2}), we obtain: $$nf \geq X \geq \frac{1}{2n} \cdot |E|^2.$$ Rearranging it completes the proof.
\end{proof}

\approxnucq*

\begin{proof} Let $H$ be an $\eps$-cut sparsifier computed by \Cref{lem:sparsifier_cq}. Then, we have $$\nu(H) = \min_{\text{ $s$-$t$ cut $(S,T)$}} |w_H(S,T)| \geq \min_{\text{ $s$-$t$ cut $(S,T)$}}  \left(1-\eps\right) \cdot |E_G(S,T)| = \left(1-\eps\right) \cdot \nu(G)$$ and
$$\nu(G) = \min_{\text{ $s$-$t$ cut $(S,T)$}}  |E_G(S,T)| \leq \min_{\text{ $s$-$t$ cut $(S,T)$}}  \left(1+\eps\right) \cdot |w_H(S,T)| = \left(1+\eps\right) \cdot \nu(H).$$ 
Hence, $f = \nu(H)$, which can be computed free of costs given $H$, satisfies the desired bounds.
\end{proof}

\ISScq*

\begin{proof} A useful intermediate tool for our analysis is an \emph{Independent-Set} (IS) Query which takes in as inputs disjoint sets $A,B \subseteq V$, and produces an output $\IS(A,B)$ which is ``YES'' or ``NO'' whether there exists an edge of $G$ among $A \times B$; i.e. whether or not $E \cap (A \times B) \ne \emptyset$. Such query can be simulated via 3 cut queries. To see this, we observe that $$|E \cap (A \times B)| = \frac{1}{2} \cdot \left(|E_G(A, V \setminus A)| + |E_G(B, V \setminus B)| - |E_G(A \cup B, V \setminus (A \cup B)) |\right).$$ Thus, $\IS(A,B)$ is equivalent to determining whether or not the RHS is 0. Each cut size can be determined via a single query (via setting $S = A,B,$ and $A \cup B$ respectively.) Therefore, we can obtain $\IS(A,B)$ via only three cut queries.

Our algorithm for answering $\FindEdge(A,B)$ proceeds as follows. We first compute $\IS(A,B)$. If the answer is ``NO'', then output ``NONE'' immediately. This only takes a single IS query to confirm that $E \cap (A \times B) = \emptyset$. Otherwise, we are assured that $E \cap (A \times B) \ne \emptyset$. In this case, we perform the binary search to find an edge $(a,b) \in E \cap (A \times B)$. Specifically, we partition $A$ into $A_1 \cup A_2$ of equal size, and $B$ into $B_1 \cup B_2$ of equal size. We then call $\IS(A_i, B_j)$ for $(i,j) \in \{1,2\}^2$, and recurse on a pair whose answer is ``YES'' until we can identify such edge. This recursion has depth $O(\log{n})$, thus we only need to make $O(\log n)$ IS queries. 
\end{proof}

\learngraphcq*

\begin{proof} Our algorithm proceeds in iteration $i = 1,2,...,z$. In iteration $i$, we wish to recover all edges in $V_i \times V_{> i}$. To do so, we first make a query $\FindEdge(V_i, V_{> i})$. If the answer comes back ``NONE'', we remove $V_i$ from the graph and proceeds to the next iteration. Otherwise, we have learned one edge in $V_i \times V_{> i}$. Then, we remove such edge from the graph, and repeat the same query again until the answer comes back ``NONE'', then proceed to the next iteration. Upon finishing iteration $i = z$, we have learned all edges of $G\langle V_1,...,V_z \rangle$.

As a result, in each iteration $i$, we only require $1 + |E(V_i, V_{>i})|$ find-edge queries. In total, this sums up to $z + |G\langle V_1,...,V_z\rangle|$ find-edge queries which can be simulated via $\tO\left(z + |G\langle V_1,...,V_z\rangle|\right)$ cut queries to $G$. We note here that each find-edge query, although not made with respect to the original graph $G$, can be simulated via $\tO(1)$ cut queries to $G$, as we have an explicit knowledge of the already-deleted edges and vertices. 
\end{proof}

\sparsifiercomm*

\begin{proof} To do so, Alice locally computes $H_A$ which is an $\eps$-cut sparsifier of $G_A$ with $\tO(n/\eps^2)$ edges and maximum weight $\tO(\eps^2 n)$. One way to do so is to exhaustively search over all integrally-weighted graphs with $\tO(n/\eps^2)$ edges and maximum weight $\tO(\eps^2 n)$, and set $H_A$ to be the one that is an $\eps$-cut sparsifier of $G_A$. Such graph must exist due to \Cref{lem:sparsifier_small_weights}, thereby must be found by Alice. Bob then does the same to compute $H_B$ w.r.t. his graph $G_B$. The player then exchange $H_A$ and $H_B$, and output $H_A \cup H_B$ as an $\eps$-cut sparsifier of $G = G_A \cup G_B$. 

It is clear that the amount of communication is bounded by $\tO(n/\eps^2)$ which is the size of $H_A$ and $H_B$. Moreover, for any cut $(S,T)$ we have
\begin{align*} 
w_{H_A \cup H_B}(S,T) & = w_{H_A}(S,T) + w_{H_B}(S,T) \\
& \leq (1+\eps) \cdot |E_{G_A}(S,T)| + (1+\eps) \cdot |E_{G_B}(S,T)| \\
& = (1+\eps) \cdot |E_{G_A \cup G_B}(S,T)|.
\end{align*}
An analogous calculation also shows that $w_{H_A \cup H_B}(S,T) \geq (1-\eps) \cdot |E_{G_A \cup G_B}(S,T)|.$ These bounds together conclude that $H_A \cup H_B$ as an $\eps$-cut sparsifier of $G = G_A \cup G_B$.
\end{proof}

\section{Spectral Sparsification with Optimal Weights}
\label{appendix:optimal_sparsifier}

For any undirected graph $G = (V,E)$ with a nonnegative weight function $w_G$, denote its \emph{Laplacian matrix} by 
$$L_G := \sum_{e \in E} w_G(e) \cdot \chi_e \chi_e^{\top}$$
whereas $\chi_e = e_a - e_b$ for an edge $e = (a,b)$.

We say a graph $H$ is an $\eps$-\emph{spectral sparsifier} of $G$ iff $E(H) \subseteq E(G)$ and for any $v \in \mathbb{R}^n$, we have

$$1-\eps \leq \frac{v^\top L_H v}{v^\top L_G v} \leq 1+\eps.$$

As a corner case, we note that any Laplacian of a connected graph has null space spanned by an all-one vector $\vec{1}$. Therefore, it suffices to consider only vectors $v$ that are orthogonal to $\vec{1}$.

\begin{observation} An $\eps$-spectral sparsifier is also an $\eps$-cut sparsifier.
\end{observation}

The goal of this section is to show that any unweighted undirected graph $G$ has an $\eps$-spectral sparsifer with an optimal tradeoff between the number of edges and the maximum weights, as stated in \Cref{lem:sparsifier_small_weights}. We begin by introducing necessary tools, and then proceed to prove the theorem.

\subsection{Preliminaries}

For any matrix $M \in \mathbb{R}^{n \times n}$, denote $\lambda_1(M) \geq \ldots  \geq \lambda_n(M)$ to be its eigenvalues in a decreasing order. We say a matrix $M$ is \emph{positive semidefinite} (or \emph{psd} for short) iff $\lambda_n(M) \geq 0$.

For any connected graph $G = (V,E)$, its Laplacian $L_G$ is symmetric and is psd. Plus, $L_G$ has $n-1$ non-zero eigenvalues, and its kernel is of dimension 1 and spanned by $\vec{1}$. Hence, we can write
$$L_G := \sum_{i = 1}^n \lambda_i u_iu_i^\top$$
where $\lambda_1 \geq \ldots \geq \lambda_{n-1} > 0$ are the top $n-1$ eigenvalues of $L_G$ and $u_1,\ldots,u_{n-1}$ are the corresponding orthonormal eigenvectors. Note that its bottom eigenvalue is $\lambda_n = 0$ with respect to an orthonormal eigenvector is $u_n = \frac{1}{\sqrt{n}} \cdot \vec{1}$.

Denote a Moore-Penrose oseudoinverse of $L_G$ as:
$$L_G^+ = \sum_{i = 1}^{n-1} \frac{1}{\lambda_i} u_iu_i^\top$$
and its ``square root'' as
$$L_G^{+/2} = \sum_{i = 1}^{n-1} \frac{1}{\sqrt{\lambda_i}} u_iu_i^\top.$$ Also denote
$$\Pi = L_G^{+/2}L_G L_G^{+/2} = \sum_{i = 1}^{n-1} u_iu_i^\top $$
to be a projection matrix on to the range of $L_G$.

\begin{observation} The matrices $L_G, L_G^{+}, L_G^{+/2}$, and $\Pi$ have the same set of orthogonal eigenvectors $u_1,\ldots,u_{n-1}$ and $\frac{1}{\sqrt{n}} \cdot \vec{1}$. They also share the same kernel which is $span(\vec{1})$. 
\end{observation}

For any edge $e$, denote $R_e := \chi_e^\top L_G^+ \chi_e$. This quantity is also known as \emph{effective resistance.} A trace argument asserts that the weighted effective resistances sums up to exactly $n-1$. This is known as the Forster's Theorem.

\begin{theorem} For a connected graph $G$, we have $\sum_{e \in G} w_eR_e = n-1$.
\label{thm:forster}
\end{theorem}

The following theorem characterizes the top and bottom eigenvalue of a matrix.
\begin{theorem}[Courant-Fischer formula] $\lambda_1(M) = \max_{v \ne \vec{0}} \frac{v^\top M v}{v^{\top}v}$ and $\lambda_n(M) = \min_{v \ne \vec{0}} \frac{v^\top M v}{v^{\top}v}$.
\label{thm:C-F}
\end{theorem}

\begin{corollary} For any vector $v \in \mathbb{R}^n$, an $n \times n$ matrix $vv^\top$ is psd. Moreover, we have $\lambda_1(vv^{\top}) = \|v\|_2^2$. 
\label{cor:max_eval}
\end{corollary}

\begin{proof} For any $x \in \mathbb{R}^n$, we have $x^\top (vv^\top) x = \|x^\top v\|_2^2 \geq 0$. Thus, $vv^\top$ is psd. Furthermore, for $x \ne 0^n$, we have 
$$\frac{x^\top (vv^\top) x}{\|x\|_2^2} = \frac{\|x^\top v\|_2^2}{\|x\|_2^2} = \frac{\langle x, v\rangle^2}{\|x\|_2^2} \leq \|v\|^2$$
due to the Cauchy-Schwartz inequality. Moreover, the equality is attained when $x = v$. This implies $\lambda_1(vv^{\top}) = \|v\|_2^2$ as wished.
\end{proof}

We will later need the following matrix concentration bounds which is a simple extension of the the matrix Chernoff bounds \cite{Tropp11}. 

\begin{theorem}[Extended Matrix Chernoff Bounds] Consider a finite sequence $\{X_k\}$ of independent random, psd symmetric matrices of dimension $n$. Let $R$ be a positive scalar such that each of the random matrix satisfies $\lambda_1(X_k) \leq R$ almost surely. Also let $\{\widetilde{X}_l\}$ be a finite of (deterministic) symmetrix psd matrices. Denote $\mathsf{X} = \sum_{i \in [k]} X_i + \sum_{j \in [l]} \widetilde{X}_j$.
Finally, define:
$$\mu_\textnormal{min} := \lambda_n( \E[\mathsf{X}]) \hspace{5mm} \text{ and } \hspace{5mm}  \mu_\textnormal{max} := \lambda_1( \E[\mathsf{X}]).$$ Then for any $\delta \in (0,1)$, we have:
$$\Pr \left[\lambda_1(\mathsf{X}) \geq (1+\delta) \mu_{\textnormal{max}}\right] \leq n \cdot \exp \left( -\frac{\delta^2 \mu_{\textnormal{max}}}{3R}\right)$$
and
$$\Pr \left[\lambda_n(\mathsf{X}) \leq (1-\delta) \mu_\textnormal{min}\right] \leq n \cdot \exp \left( -\frac{\delta^2 \mu_\textnormal{min}}{2R}\right).$$
\label{thm:extended_matrix_chernoff} 
\end{theorem}

We make a few remarks regarding the theorem. First of all, unlike the original statement of \cite{Tropp11}, here we include a set of deterministic matrices $\{\widetilde{X}_l\}$. Secondly, we will eventually apply such theorem with a set of matrices that yields $\E[\mathsf{X}] = \Pi$ whose eigenvalues are $1$ with multiplicity $n-1$ and $0$ with multiplicity $1$ (w.r.t. eigenvector $\vec{1}$). Since it suffices to only work with vectors orthogonal to $\vec{1}$, we can treat $\Pi$ as having all eigenvalues $1$.\footnote{To see this, let $\Pi' := \Pi + \frac{1}{n} \cdot \vec{1}\vec{1}^\top$ so that it has the same set of eigenvectors as $\Pi$ but now all eigenvalues are $1$s. Then, for any vector $v \perp \vec{1}$, we have $v^\top \Pi' v = v^\top \Pi v$.} These assumptions we made do not affect the integrity of such theorem nor of our proofs.

\subsection{Proof of \Cref{lem:sparsifier_small_weights}}

We first show a main lemma which will be used recursively.

\begin{lemma} There exists a universal constant $C > 0$ such that the following is true. Let $G$ be a weighted undirected graph with $n$ vertices, $m$ edges, and edge weights bounded by $W$. Let $\lambda$ be such that $\omega(n^{-\frac{1}{2}}) \leq \lambda \leq o(1)$. There is an $\sqrt{\frac{Cn \log n}{\lambda m}}-$spectral sparsifier of $G$ with $m \cdot \left(\frac{1}{2} \pm \lambda)\right)$ edges and edge weights bounded by $2W$. 
\label{clm:one_shot_sparsifier}
\end{lemma}

\begin{proof} Consider the following random procedure.
\begin{enumerate}
    \item Let $E_{\text{low}} \subseteq E$ be a set of edges for which $w_eR_e \leq \frac{n}{\lambda m}$. Also let $E_{\text{high}} = E \setminus E_{\text{low}} $. Then due to \Cref{thm:forster}, we have $|E_{\text{low}}| \geq (1-\lambda)m$.
    \item For each $e \in E_{\text{low}}$, draw $z_e \in \{0, 1\}$ uniformly at random.
    \item Let $Y$ be a graph consisting of $e \in E_{\text{low}}$ such that $z_e = 1$ and double the weight to $2w_e$.
    \item Output $H = Y \cup E_{\text{high}}$.
\end{enumerate}

We claim that the output graph $H$ is the desired sparsifier with non-zero probability. First of all, let's count number of edges in $H$. This is $$|E_{\text{high}}| + B(|E_{\text{low}}|, \frac{1}{2}) = m - |E_{\text{low}}| + B(|E_{\text{low}}|, \frac{1}{2}) \equiv m - B(|E_{\text{low}}|, \frac{1}{2}).$$
Via Chebychev inequality, we have the followings with constant probability.
$$B(|E_{\text{low}}|, \frac{1}{2}) \leq \frac{|E_{\text{low}}|}{2} + 4\sqrt{|E_{\text{low}}|} \leq \frac{m}{2} + 4\sqrt{m} \leq m \cdot (\frac{1}{2}+\lambda)$$
and 
$$B(|E_{\text{low}}|, \frac{1}{2}) \geq \frac{|E_{\text{low}}|}{2} - 4\sqrt{|E_{\text{low}}|} \geq \frac{(1-\lambda)m}{2} - 4\sqrt{m} \geq m \cdot (\frac{1}{2}-\lambda).$$
Hence, with constant probability, the number of edges of $H$ lies within $m \cdot \left(\frac{1}{2} \pm \lambda \right)$.

Next, we will analyze the approximation error. For any $e \in E_{\text{low}}$, denote a random vector $y_e := \sqrt{2w_ez_e} \cdot L^{+/2}_G\chi_e$ and an $n \times n$ random matrix $Y_e := y_ey_e^\top$. Similarly, for any $e \in E_{\text{high}}$, denote a vector $y_e := \sqrt{w_e} \cdot L^{+/2}_G\chi_e$ and an $n \times n$ matrix $Y_e := y_ey_e^\top $. Notably, randomness is \emph{not} involved in those $Y_e$ of  $e \in E_{\text{high}}$ and every $Y_e$ is psd and symmetric. 

By basic calculations, we can show that $ \sum_{e \in E} Y_e = L^{+/2}_G L_H L^{+/2}_G$, $\E\left(\sum_{e \in E} Y_e\right) = \Pi$, and $\lambda_1(Y_e) \leq \frac{2n}{\lambda m}$ for every $e \in E_{\text{low}}$. We defer these calculations to via \Cref{fact:sum_ys} and \Cref{fact:bound_eig_Y} which will come up shortly. Applying the extended matrix Chernoff bound (\Cref{thm:extended_matrix_chernoff}), the random matrix $L_G^{+/2} L_H L_G^{+/2}$ has all non-zero eigenvalues between $1 \pm \eps$ except probability $2n \cdot \exp(-\frac{\eps^2 \lambda m}{6n}).$
By setting $\eps = \sqrt{\frac{Cn \log n}{\lambda m }}$ for suitably large constant $C$, such error probability is at most $n^{-10}$. 

When it is the case that $L_G^{+/2} L_H L_G^{+/2}$ has all eigenvalues lied between $1 \pm \eps$, we claim that $H$ is a spectral sparsifier of $G$. To see this, take any vector $v \in \mathbb{R}^n$ and let $u = L_G^{+/2}v$. If $u = \vec{0}$, it means $v = \vec{1}$ which appears to be the corner case.\footnote{$v = \vec{1}$ is in the kernel of both $L_G$ and $L_H$. As a result, we have $v^\top  L_G  v = v^\top  L_H  v = 0$ which satisfies the bounds $1-\eps \leq \frac{v^\top  L_H  v}{v^\top  L_G v} \leq 1+\eps$ if we were to abuse that $\frac{0}{0} = 1$.} Otherwise, we have:

$$\frac{v^\top  L_H  v}{v^\top  L_G v} = \frac{u^\top  L_G^{+/2} L_H L_G^{+/2} u}{u^\top  u} \leq \lambda_1(L_G^{+/2} L_H L_G^{+/2}) \leq 1+\eps$$
and 
$$\frac{v^\top  L_H  v}{v^\top  L_G v} = \frac{u^\top  L_G^{+/2} L_H L_G^{+/2} u}{u^\top  u} \geq \lambda_n(L_G^{+/2} L_H L_G^{+/2}) \geq 1-\eps$$
where both final inequalities follows the Courant-Fischer formula. 

By the union bound, both conditioning events (bounds on the number of edges and on the eigenvalues) occur with non-zero probability. This means there exists $H$ with $m \cdot \left(\frac{1}{2} \pm \lambda\right)$ edges that is a  $\sqrt{\frac{Cn \log n}{\lambda m}}$-spectral sparsifier of $G$. Finally, each edge of $H$ has its weight at most doubled up from $G$. Thus, the maximum weight of $H$ is at most $2W$.
\end{proof}

\begin{fact} $\sum_{e \in E} Y_e = L^{+/2}_G L_H L^{+/2}_G$ and $\E \left(\sum_{e \in E} Y_e \right)= \Pi$.
\label{fact:sum_ys}
\end{fact}
\begin{proof} Consider the following derivation:
    \begin{align*}
    \sum_{e \in E} Y_e & = \sum_{e \in E_{\text{low}}} Y_e + \sum_{e \in E_{\text{high}}} Y_e \\
    & = \sum_{e \in E_{\text{low}}} L^{+/2}_G \left(2w_ez_e \cdot \chi_e \chi_e^\top \right) L^{+/2}_G + \sum_{e \in E_{\text{high}}} L^{+/2}_G \left(w_e \chi_e \chi_e^\top \right) L^{+/2}_G\\
    & = L^{+/2}_G \left(\sum_{e \in E_{\text{low}}} 2w_ez_e \cdot  \chi_e \chi_e^\top  + \sum_{e \in E_{\text{high}}} w_e \chi_e \chi_e^\top  \right) L^{+/2}_G  \\
    & = L^{+/2}_G L_H L^{+/2}_G.
\end{align*}
This proves the first equality.

For the second equality, notice that $\E(L_H) = L_G$ since each $e \in E_{\text{low}}$ has equal probabilities of its weight being doubled or zeroed. As a result, we have:
$$\E\left(\sum_{e \in E} Y_e\right) = \E\left(L^{+/2}_G L_H L^{+/2}_G\right) = L^{+/2}_G \E\left(L_H\right) L^{+/2}_G = L^{+/2}_G L_G L^{+/2}_G = \Pi$$
as wished. \end{proof}

\begin{fact}
    For every $e \in E_{\textnormal{low}}$, we have $\lambda_1(Y_e) \leq \frac{2n}{\lambda m}$.
\label{fact:bound_eig_Y}
\end{fact}
\begin{proof}
Following \Cref{cor:max_eval}, we have
\begin{align*}
   \lambda_1(Y_e) = \|y_e\|^2 & = 2w_ez_e \cdot \| L^{+/2}_G\chi_e \|_2^2 \leq 2w_e \cdot \chi_e^\top  L_G^+ \chi_e = 2w_eR_e \leq \frac{2n}{\lambda m}
\end{align*}
where the first inequality is due to the fact that $z_e \in \{0, 1\}$ and the second ineqality is due to the definition of $E_{\text{low}}$.
\end{proof}

Now we are ready to prove an analogue of \Cref{lem:sparsifier_small_weights} for spectral sparsifier which is inherently a stronger result.

\begin{theorem} Let $G$ be an unweighted undirected connected graph with $n$ vertices and average degree $d$. Then, for any $\eps \in (0,1)$, there is an $\eps$-spectral sparsifier of $G$ with $O\left(\frac{n \log^2{n}}{\eps^2}\right)$ edges and edge weights bounded by $O\left(\frac{\eps^2d}{\log^2{n}}\right)$.
\end{theorem} 

\begin{proof} Set $\lambda = (c \cdot \log n)^{-1}$ for sufficiently large constant $c$. Denote $G_0$ by $G$. For each $i$, let $G_i$ denote an $\eps_i$-spectral sparsifier of $G_{i-1}$ obtained from \Cref{clm:one_shot_sparsifier} with respect to $\lambda$. Let $m_i$ denote the number of edges of $G_i$. Via \Cref{clm:one_shot_sparsifier}, for each $i$, we have $m\cdot \left(\frac{1}{2}-\lambda\right)^i  \leq m_i \leq m \cdot \left(\frac{1}{2}+\lambda\right)^i$ and $\eps_i = \sqrt{\frac{Cn \log^2 n}{m_i}}$ for some universal constant $C>0$. Let $\ell$ be the largest non-negative integer such that  
$$ \sqrt{\frac{n \log^2 n}{ m}} \cdot \left(\frac{1}{2} - \lambda\right)^{-\ell/2} \leq \frac{\eps}{40 \sqrt{C}}.$$ 
We claim that $G_\ell$ is the desired $\eps$-spectral sparsifier of $G$. Note here that $\ell = O(\log n)$.

\paragraph{Maximum weight of $G_\ell$.} Notice that maximum weight of $G_{\ell}$ is $2^\ell$, and we have
\begin{align*}
    2^{\ell} \leq \left(\frac{1}{2} - \lambda\right)^{-\ell} \leq O\left( \frac{\eps^2 m}{n \log^2 n} \right) = O\left( \frac{\eps^2 d}{\log^2 n} \right).
\end{align*}

\paragraph{Number of edges of $G_\ell$.}  By definition of $\ell$, we have 
$\sqrt{\frac{n \log^2 n}{ m}} \cdot \left(\frac{1}{2} - \lambda\right)^{-(\ell+1)/2} > \frac{\eps}{40\sqrt{C}}.$ This further implies $\left(\frac{1}{2} - \lambda\right)^{\ell} < 400C \cdot \frac{n \log^2 n}{\eps^2 m}$. We then have:
\begin{align*}
    m_\ell  \leq m \cdot \left(\frac{1}{2} + \lambda\right)^{\ell} & = m \cdot \left(\frac{1}{2} - \lambda\right)^{\ell} \cdot \left(\frac{1+2\lambda}{1-2\lambda}\right)^\ell \\
    & \leq \frac{n \log^2 n}{\eps^2} \cdot \left(400C \cdot e^{6\lambda \ell}\right) \tag{$\frac{1+2\lambda}{1-2\lambda} \leq 1+6\lambda \leq e^{6\lambda}$}\\
    & = O\left(\frac{n \log^2 n}{\eps^2 }\right) \tag{$\lambda = O(\frac{1}{\log{n}})$ and $\ell = O(\log{n})$}
\end{align*}

\paragraph{Approximation factor.} Let us first bound $\eps_1+\ldots + \eps_\ell$. Consider the following derivation.

\begin{align*}
    \sum_{i \in [\ell] }\eps_i = \sum_{i \in [\ell] } \sqrt{\frac{Cn \log^2 n}{m_i}} & = \sqrt{\frac{Cn \log^2 n}{m}} \cdot \sum_{i \in [\ell] } \left(\frac{1}{2}-\lambda\right)^{-i/2} \\
    & = \sqrt{\frac{Cn \log^2 n}{m}} \cdot 5 \cdot \left(\frac{1}{2}-\lambda\right)^{-\ell/2} \tag{geometric sum} \\
    & \leq \sqrt{\frac{Cn \log^2 n}{m}} \cdot 5 \cdot \left(\frac{\eps}{40\sqrt{C}} \cdot \sqrt{\frac{n \log^2 n}{m}}\right) \tag{definition of $\ell$}\\
    & \leq \frac{\eps}{8}.
\end{align*}

Next, we compare $v^\top L_{G_\ell}v$ and $v^\top L_{G}v$ for any vector $v \in \mathbb{R}^n$. Again, we omit the corner case of $v = \vec{1}$. Then, we have:
\begin{align*}
    \frac{v^\top L_{G_\ell} v}{v^\top L_{G}v} \leq \prod_{i \in [\ell]} (1+\eps_i) \leq e^{\sum_{i \in [\ell]} \eps_i} \leq e^{\eps/8} \leq 1+\eps.
\end{align*}
and 
\begin{align*}
    \frac{v^\top L_{G_\ell} v}{v^\top L_{G}v} \geq \prod_{i \in [\ell]} (1-\eps_i) \geq 1- \sum_{i \in [\ell]} \eps_i \geq 1-\frac{\eps}{8} \geq 1-\eps.
\end{align*}
which concludes that $G_\ell$ is an $\eps$-spectral sparsifier of $G$ as wished.
\end{proof}

\section{Finding a Large Flow via Communication: Improved Analysis}
\label{sec:large_flow_comm}

The goal of this section is to prove \Cref{lem:large_flow_comm}, thereby concluding our minimum $s$-$t$ cut algorithm in the two-player communication model. We first state our main claim.

\begin{restatable}{clm}{mainclaimcomm} Let $f,k,k'$ be known parameters such that $k' < k \leq f$ with a promise that $\nu(G) \geq f$, and also given explicitly is a subgraph $H \preceq G$ with $\nu(H) \geq f-k$. Then, there is a deterministic protocol which communicates $\tO\left(nf\cdot\frac{\sqrt{k}}{k'}\right)$ bits and outputs $H^{\textnormal{after}} \preceq G$ with $\nu(H^{\textnormal{after}}) \geq f-k'$.
\label{claim:main_claim_comm}
\end{restatable}

As a corollary, we can obtain a ``promise'' version of \Cref{lem:large_flow_comm}.

\begin{corollary} Let $f$ and $\Delta \leq f$ be parameters with a promise that $\nu(G) \geq f$.Then, there is a deterministic protocol which communicates $\tO\left(\frac{nf}{\sqrt{\Delta}}\right)$ bits and outputs $H \preceq G$ with $\nu(H) \geq f-\Delta$.
\label{cor:main_lem_comm}
\end{corollary} 

\begin{proof} Use \Cref{claim:main_claim_comm} recursively with $(k,k') = (f, f/2),\ldots,(2\Delta, \Delta)$ with $H$ being $H^{\text{after}}$ from the previous iteration (for the first iteration, set $H = \emptyset$.) Inductively, each iteration is successful to obtain $\nu(H^{\text{after}}) \geq f-k'$. Thus, upon the completion we obtain $H$ with $\nu(H) \geq f-\Delta$. In total, the number of cut-queries made to $G$ is $\tO\left(\frac{nf}{\sqrt{\Delta}}\right)$ due to the geometric sum.
\end{proof}

Now let's prove \Cref{lem:large_flow_comm} from \Cref{cor:main_lem_comm}

\largeflowcomm*

\begin{proof} Consider the following deterministic communication protocol.

\begin{enumerate}
    \item Compute $\hat{f}$ such that $0.99 \nu(G) \leq \hat{f} \leq 1.01 \nu(G)$ using $\tO(n)$ bits via \Cref{lem:sparsifier_comm}.
    \item If $\hat{f} \leq 0.99 \Delta$, output $H := \emptyset$.
    \item Else, we have $\hat{f} \geq 0.99 \Delta$. Set $\lambda = \frac{\Delta}{20\hat{f}} \leq \frac{1}{10}$.
    \item Compute $f$ such that $(1-\lambda)\nu(G) \leq f \leq (1+\lambda)\nu(G)$ using $O(n/\lambda^2)$ bits via \Cref{lem:sparsifier_comm}.
    \item If $f \leq (1-\lambda) \Delta$, output $H = \emptyset$.
    \item If $f \geq (1-\lambda) \Delta$, use  \Cref{cor:main_lem_comm} with $(f, \Delta) \leftarrow \left(\frac{f}{1+\lambda}, \frac{\lambda f}{10}\right)$ and output such $H$.
\end{enumerate}

Let us first argue the correctness. In Step 2, if $\hat{f} \leq 0.99 \Delta$, we then know that $\nu(G) \leq \Delta$; hence, simply outputting $H = \emptyset$ is correct.

Otherwise, we must have $\hat{f} \geq 0.99 \Delta$. In Step 4, we compute $f$. Similar to before, if $f \leq (1-\lambda) \Delta$, we then know that $\nu(G) \leq \Delta$; hence, simply outputting $H = \emptyset$ is correct.

Otherwise, we must have $f \geq (1-\lambda) \Delta$. We also know that  $\nu(G) \geq \frac{f}{1+\lambda}$. Together with the fact that $\frac{f}{1+\lambda}  \geq \frac{\lambda f}{10} $, the premises of \Cref{cor:main_lem_comm} are met. As a result the output $H$ via Step 6 must satisfy:
    \begin{align*}
         \nu(H) \geq \frac{f}{1+\lambda} -\frac{\lambda f}{10} & = \nu(G) - \left[\nu(G) - \left(\frac{1}{1+\lambda} - \frac{\lambda}{10}\right) \cdot f \right] \\
         & \geq  \nu(G) - \left(\frac{1}{1-\lambda} - \frac{1}{1+\lambda} + \frac{\lambda}{10}\right) \cdot f \tag{$\nu(G) \leq \frac{f}{1-\lambda}$} \\
         & \geq \nu(G) - 10 \lambda f \tag{$\lambda \leq \frac{1}{10}$} \\
         & \geq \nu(G) - \frac{\Delta}{2} \cdot \frac{f}{\hat{f}} \tag{$\lambda = \frac{\Delta}{20 \hat{f}}$} \\
         & \geq \nu(G) - \Delta
    \end{align*}
where the last inequality follows the facts that $f \leq (1+\lambda) \nu(G) \leq 1.1 \nu(G)$ and $0.99 \nu(G) \leq \hat{f}$. Hence, our protocol does compute $H$ with a guarantee that $\nu(H) \geq \nu(G) - \Delta$.

Finally, the protocol communicates $\tO(n)$ bits in Step 1, $\tO\left(\frac{n}{\lambda^2}\right)$ bits in Step 4, and $\tO\left(\frac{n\sqrt{f}}{\sqrt{\lambda}}\right)$ bits in Step 6. Using the fact that $f \leq (1+\lambda)\nu(G) \leq 1.1 \nu(G)$ and $\lambda = \frac{\Delta}{20\hat{f}} \geq \frac{\Delta}{20 (1.01) \cdot \nu(G)}$, the total communication is bounded by $$\tO\left(\frac{n \cdot \nu(G)}{\sqrt{\Delta}} + \frac{n \cdot \nu(G)^2}{\Delta^2}\right)$$ bits, as desired.
\end{proof}

\paragraph{Setting up the proof of \Cref{claim:main_claim_comm}.}
\label{sec:main_lem}

Our task only remains to \Cref{claim:main_claim_comm}. We assume its setting: let $f,k,k'$ be such that $k' < k \leq f$ with a promise that $\nu(G) \geq f$, and also given for free is $H \preceq G$ with $\nu(H) \geq f-k$. We wish to obtain $H^{\text{after}} \preceq G$ with $\nu(H^{\text{after}}) \geq f-k'$. Our protocol proceeds in a similar fashion as sketched in \Cref{sec:setup} barring changes in the notions of witness: we shall now consider a witness with respect to the residual graph instead of actual graph $H$.

Specifically, let $F$ be a flow subgraph of $H$ with size $\nu(F) = f - k$ that is fixed throughout the protocol. For any graph $\mathcal{G}$ containing $F$ as a subgraph, we define a set of \emph{residual witness} of $\mathcal{G}$ with respect to $F$, denoted $\mathcal{W}_F(\mathcal{G})$, with an invariant that $\nu(\mathcal{G}) \geq f - k'$ if and only if $\mathcal{W}_F(\mathcal{G}) = \emptyset$. The protocol starts with $H_1 \leftarrow H$. In each round $i = 1, 2, \ldots$, the players compute a private edge $e_i \in G \setminus H_i$ that reduces the number of witnesses by a fraction of $\gamma$; that is, $e_i$ be such that $|\mathcal{W}_F(H_i \cup {e_i})| \leq (1-\gamma) \cdot |\mathcal{W}_F(H_i)|$, exchange such edge, and then set $H_{i+1} \leftarrow H_i \cup {e_i}$. The protocol terminates at the first round $T$ where $\mathcal{W}_F(H_T) = \emptyset$, ensuring that the final graph $H^{\text{after}} := H_T$ has $\nu(H^{\text{after}}) \geq f-k'$ via our invariant. Since the protocol communicates $T$ edges in total, our goal is to ensure that a small $T= \widetilde{O}\left(nf \cdot \frac{\sqrt{k}}{k'}\right)$ suffices for termination.

\paragraph{Organization.} In \Cref{subsec:preprocess}, we apply an inexpensive pre-processing step to $H$, ensuring that all of its vertices is actually reachable from the terminal $s$. In \Cref{subsec:witness}, we formally define \emph{residual witnesses} and prove their corresponding useful properties. \Cref{subsec:witness_ub} establishes an upper bound on the starting number of residual witnesses of $H$. Finally, \Cref{subsec:conclude} integrates these results, yielding a protocol that satisfies \Cref{claim:main_claim_comm}.

\subsection{Pre-processing $H$} 
\label{subsec:preprocess}

As an initial step, we apply a pre-processing procedure to $H$ as follows:

\begin{enumerate}
\item Alice computes a spanning forest $\mathcal{F}_A$ of her full input graph $G_A = (V,E_A)$. She then sends $\mathcal{F}_A$ to Bob and adds it to $H$.
\item Bob computes a spanning forest $\mathcal{F}_B$ of his full input graph $G_B = (V,E_B)$. He then sends $\mathcal{F}_B$ to Alice and adds it to $H$.
\item The players then combine $\mathcal{F}_A \cup \mathcal{F}_B$ and compute a set of vertices $V' \subseteq V$, consisting of vertices in $V$ that are reachable from $s$ via $\mathcal{F}_A \cup \mathcal{F}_B$.
\item Alice and Bob remove all vertices in $V \setminus V'$ from their respective private graphs and from the public graph $H$.
\end{enumerate}

This pre-processing procedure is relatively cheap in communication: it requires only 2 rounds of communication and incurs a cost of $\widetilde{O}(n)$ bits. A key observation is that the graph $\mathcal{F}_A \cup \mathcal{F}_B$ preserves pairwise connectivity of vertices in $G$. That is, a pair of vertices $u$ and $v$ is connected in $G$ if and only if they are connected in $\mathcal{F}_A \cup \mathcal{F}_B$.  Therefore, the vertex set $V'$ obtained in Step 3 represents all the vertices in $V$ that are reachable from $s$ in $G$.

The consequences of this pre-processing are as follows. Steps 1 and 2 only add edges to $H$, ensuring that $\nu(H)$ never decreases. The deletions in Step 4 do not affect the max-flow value of either $G$ or $H$, because any vertex not reachable from $s$ cannot participate in any $s$-$t$ flow. Thus, the pre-conditions of \Cref{claim:main_claim_comm} remain valid after the pre-processing.

Moreover, after the pre-processing, the vertex set of $H$ becomes $V'$, and $H$ contains $\mathcal{F}_A \cup \mathcal{F}_B$ as a subgraph. Consequently, every vertex in $H$ is reachable from $s$, enabling us to assert the following assumption.

\begin{assumption}
In the public graph $H$, every vertex $v \in V(H)$ is reachable from $s$.
\label{assumption:H_connected}
\end{assumption}

\subsection{Residual Witness}
\label{subsec:witness}

We now formally define the revised notion of witness in which we call a \emph{residual witness}.

\begin{definition}[Residual Witness]
Let $\mathcal{G} = (V, E)$ be an unweighted undirected graph with two designated terminals $s$ and $t$, and $G$ contains $F$ as a flow subgraph. A tuple $\mathsf{W} = (S, T, Y)$ is called a \emph{residual witness} of $\mathcal{G}$ w.r.t $F$ (or just residual witness for short) if the following conditions hold:  

\begin{enumerate}  
    \item $(S, T)$ is an $s$-$t$ cut of $\mathcal{G}_F$ with value $w_{\mathcal{G}_F}(S, T) \leq k - k' - 1$.  
    \item $Y = E_{\mathcal{G}_F}(S,T) \cup X$, where $X \subseteq (S \times T)\setminus E_{\mathcal{G}_F}(S,T)$ is of size at most $k-k'-1-w_{\mathcal{G}_F}(S, T)$.
\end{enumerate}

\label{def:residual_witness}
\end{definition}

Intuitively, each $(a, b) \in X$ represents a \emph{currently non-existent} unit-capacity undirected edge that cross the cut $(S, T)$. The set $X$ then represents a \emph{pre-allocated} set of edges that could potentially be added to $\mathcal{G}$ in the future, ensuring that the cut value of $(S, T)$ remains at most $k - k' - 1$. Note unlike the previous definition of witness, it is possible to have $X = \emptyset$.

\begin{restatable}[Residual Witness Invariant]{inv}{invariantres} For any $\mathcal{G}$ containing $F$ as a flow subgraph, we have
$\nu(\mathcal{G}) \geq f - k'$ if and only if $\mathcal{W}_F(\mathcal{G}) = \emptyset$.
\label{invariant_res}
\end{restatable}

\begin{proof}
First, suppose that $\nu(\mathcal{G}) \geq f - k'$. Assume for the sake of contradiction that a residual witness exists and let $(S,T)$ be a cut with respect to that witness. By definition of a witness, the cut $(S, T)$ has value $w_{\mathcal{G}_F}(S, T) \leq k - k' - 1$. Via \Cref{clm:reduce_cut_by_flow}, we have $|E_{\mathcal{G}}(S,T)| = \nu(F) + w_{\mathcal{G}_F}(S, T)$. Thus, we derive:  $$f - k' \leq \nu(\mathcal{G}) \leq |E_{\mathcal{G}}(S,T)| = \nu(F) + w_{\mathcal{G}_F}(S, T) \leq (f-k) + (k - k' - 1) = f - k' - 1$$
yielding a contradiction. This concludes $\mathcal{W}_F(\mathcal{G}) = \emptyset$.

Conversely, suppose that $\nu(\mathcal{G}) \leq f - k' - 1$. Let $(S,T)$ be a minimum $s$-$t$ cut of $\mathcal{G}$ with value $|E_{\mathcal{G}}(S,T)| = \nu(G) \leq f - k' - 1$. Again, following \Cref{clm:reduce_cut_by_flow}, we have
$$w_{\mathcal{G}_F}(S,T) = |E_{\mathcal{G}}(S,T)| - \nu(F) \leq (f - k' - 1) - (f-k) = k-k'-1.$$
Then by taking $X = \emptyset$, $(S, T, E_{\mathcal{G}_F}(S,T))$ is a residual witness of $\mathcal{G}$, implying $\mathcal{W}_F(\mathcal{G}) \neq \emptyset$.
\end{proof}

We now define the process of eliminating a witness. This is analogous to \Cref{def:kill}.

\begin{restatable}[Killing a residual witness]{defi}{killres}
Let $\mathcal{G} = (V,E)$ be an unweighted and undirected graph with terminals $s$ and $t$. Let $F$ be its flow subgraph. Let $\mathsf{W} = (S, T, Y)$ be a residual witness of $\mathcal{G}$. We say that a pair of vertices $(a, b) \notin E$ \emph{kills} $\mathsf{W}$ iff $\mathsf{W}$ is \emph{not} a residual witness of $\mathcal{G} \cup (a,b)$. Equivalently, this occurs if and only if $a \in S$, $b \in T$, and $(a, b) \notin Y$. 
\end{restatable}

Recall the brief description of our protocol that Alice and Bob begin with $H_1 := H$ and iteratively add private edge(s) that eliminate a $\gamma$ fraction of remaining residual witnessed. The next claim shows that when $\gamma = \Omega\left(\frac{k'}{n\sqrt{k}}\right)$, at least one player is in possession of such edge; thus, by having each player exhaustively search for that edge, it is surely found.

\begin{claim}
    For any iteration $i$ where $\mathcal{W}_F(H_i) \ne \emptyset$, there exists an edge $e_i \in G \setminus H_i$ such that 
    $$|\mathcal{W}_F(H_i \cup e_i)| \leq \left(1-\Omega\left(\frac{k'}{n\sqrt{k}}\right)\right) \cdot |\mathcal{W}_F(H_i)|.$$
    \label{claim:dec_witnesses}
\end{claim}

\begin{proof} 
Recall that $F$ is a flow graph of size $f - k$, and $G_F$ is the residual graph of $G$ with respect to $F$. By \Cref{prop:sum_flows}, $G_F$ has a maximum flow of at value least $k$ since  
$$\maxflow(G_F) = \maxflow(G) - \maxflow(F) \geq f - (f - k) = k.$$

Let $Q$ be a flow subgraph of $G_F$ that consists of $k$ non-circular flows. By \Cref{lem:flow_cover}, $Q$ contains $O(n\sqrt{k})$ edges. Let $Q^* \subseteq Q$  denote the set of edges in $Q$ that does not appear in $(H_i)_F$. Then, we have $Q \preceq (H_i)_F \cup Q^*$, implying that $(H_i)_F \cup Q^*$ has a maximum flow, i.e., a minimum cut, of value at least $k$.

Let $\mathsf{W} = (S, T, E_{(H_i)_F}(S,T) \cup X)$ be an arbitrary residual witness of $H_i$. We know that $w_{(H_i)_F \cup Q^*}(S, T) \geq k$ while  
$w_{(H_i)_F \cup X}(S, T) = w_{(H_i)_F}(S,T) + |X| \leq k - k' - 1.$\footnote{Here, we slightly abuse the notation by treating $X$ as a set of undirected unit-capacitated edges in $(H_i)_F$, even though they do not actually exist in $(H_i)_F$.} Thus, we obtain $w_{Q^* \setminus X}(S, T) \geq k' + 1$. In addition, we have $Q^*\setminus X \subseteq Q \subseteq G \setminus H_i$ only consists of private undirected edge with weight 1. Hence, the set $Q^* \setminus X$ represents (at least) $k' + 1$ edges among $Q^*$ that kills $\mathsf{W}$ . 

Finally, by the averaging argument over all residual witnesses, there must exist a private edge in $Q^*$ that kills at least a fraction of $$\frac{k' + 1}{|Q^*|} \geq \frac{k' + 1}{|Q|} = \Omega\left(\frac{k'}{n \sqrt{k}}\right)$$
of $\W_F(H_i)$. This completes the proof.
\end{proof}

\subsection{Bounding the Number of Residual Witnesses of $H$}
\label{subsec:witness_ub}

We now bound the number of residual witnesses of $H$ with respect to its flow subgraph $F$ of size $f - k$. Throughout this subsection, we write $(a, b, w)$ to denote a \emph{directed} edge from $a$ to $b$ with weight $w$, and reserve the notation $(a, b)$ exclusively for an undirected, unweighted edge between vertices $a$ and $b$. In the context of cuts and flows, we may replace an undirected edge $(a, b)$ with the pair of directed edges $(a, b, 1)$ and $(b, a, 1)$ as needed. We adopt these views interchangeably.

\begin{claim} The number of residual witnesses of $H$ is at most $n^{O(f)}$.
\label{claim:witnesses_ub}
\end{claim}

\begin{proof}
We will attempt to (over)count the number of residual witnesses $\mathsf{W} = (S, T, E_{H_F}(S,T) \cup X)$ of $H$. First, observe that since $E_{H_F}(S,T)$ contains up to $k - k' - 1 \leq k \leq f$ edges, there are at most $n^{O(f)}$ choices for $E_{H_F}(S,T)$. Similarly, as $|X| \leq k - k' - 1 \leq k \leq f$, there are at most $n^{O(f)}$ choices for $X$. In what follows, we will show that for any fixed realization of $C \subseteq S \times T \times \{1,2\}$ with $w(C) \leq k-k'-1$, the number of cuts $(S, T)$ such that $E_{H_F}(S,T) = C$ is at most $n^{O(f)}$. This count is equivalent to determining the number of ways to label each $v \in V$ with either `s' or `t' such that the set of directed edges from vertices labeled `s' to those labeled `t' corresponds exactly to $C$.  

For convenience, we adopt the following terminology: a \emph{labeling} is a function $\phi: V \rightarrow \{\text{`s'}, \text{`t'}\}$ such that $\phi(s) = \text{`s'}$ and $\phi(t) = \text{`t'}$. A directed edge $(a, b, w) \in H_F$ is called an \emph{$s$-$t$ edge} if $\phi(a) = \text{`s'}$ and $\phi(b) = \text{`t'}$. A labeling $\phi$ is said to be \emph{consistent} with $C$ if the set of $s$-$t$ edges is precisely $C$. Therefore, our task reduces to bounding the number of consistent labelings.  

The edges of $H_F$ can be partitioned into two disjoint sets, $E_1 \cup E_2$, defined as follows:
\begin{itemize} 
\item $E_1 = H \setminus F$ is the set of undirected unit-capacitated edges. Each such edge $(a,b)$ can equivalently be replaced with two directed edge $(a,b,1)$ and $(b,a,1)$. We adopt these views interchangeably. Let $V(E_1)$ denote the set of vertices incident to edges in $E_1$. 

\item $E_2$ is the set of directed edges with weight 2, corresponding to the reversal of the flow $F$. Let $V(E_2)$ denote the set of vertices incident to edges in $E_2$. 
\end{itemize}

We now proceed to count the number of labels $\phi$ that are consistent with $C$. Our approach is to assign the labels in three steps: to $V(C)$, to $V(E_2) \setminus V(C)$, and to the remaining vertices. By \Cref{assumption:H_connected}, we may assume that $H$ is connected. Since $H$ and $H_F$ share the same underlying graph, we have $V(E_1) \cup V(E_2) = V$. Consequently, assigning a label to each vertex in $V(E_1) \cup V(E_2)$ ensures that no vertex remains unlabeled.

\paragraph{Step 1: Counting the number of labels for $V(C)$.}
Fix $C$. There can be at most one labeling of $V(C)$ that is consistent with $C$; specifically, for any $(a, b, w) \in C$, we must have $\phi(a) = \text{`s'}$ and $\phi(b) = \text{`t'}$.  See \Cref{fig:step1} for an illustration. Going forward, we fix the labels of the vertices in $V(C)$ accordingly.

\begin{figure}[h]
    \centering
    \begin{tikzpicture}

    \node[draw, circle, minimum size = 8mm, fill=gray!20] (s) at (-2.5, -1) {$s$};

    \node[draw, circle, minimum size = 8mm] (A0) at (0, 2) {};
    \node[draw, circle, minimum size = 8mm] (A1) at (0, 0) {};
    \node[draw, circle, minimum size = 8mm] (A2) at (0, -2) {`t'};
    \node[draw, circle, minimum size = 8mm] (A3) at (0, -4) {};

    \node[draw, circle, minimum size = 8mm] (B1) at (2.5, 2) {};
    \node[draw, circle, minimum size = 8mm] (B2) at (2.5, 0) {};
    \node[draw, circle, minimum size = 8mm] (B3) at (2.5, -2) {`s'};
    \node[draw, circle, minimum size = 8mm] (B4) at (2.5, -4) {`s'};

    \node[draw, circle, minimum size = 8mm] (C1) at (5, 2) {`t'};
    \node[draw, circle, minimum size = 8mm] (C2) at (5, 0) {`t'};
    \node[draw, circle, minimum size = 8mm] (C3) at (5, -2) {`s'};
    \node[draw, circle, minimum size = 8mm] (C4) at (5, -4) {};

    \node[draw, circle, minimum size = 8mm] (D1) at (7.5, 2) {`s'};
    \node[draw, circle, minimum size = 8mm] (D2) at (7.5, 0) {};
    \node[draw, circle, minimum size = 8mm] (D3) at (7.5, -2) {`s'};
    \node[draw, circle, minimum size = 8mm] (D4) at (7.5, -4) {};

    \node[draw, circle, minimum size = 8mm, fill=gray!20] (t) at (10, -1) {$t$};

    \draw[<-, >=stealth, blue] (s) -- (A1);
    \draw[<-, >=stealth, blue] (A1) -- (B2);
    \draw[<-, >=stealth, blue] (B2) -- (B1);
    \draw[<-, >=stealth, blue] (B1) -- (C1);
    \draw[<-, >=stealth, blue, snake] (C1) -- (D1);
    \draw[<-, >=stealth, blue] (D1) -- (t);

    \draw[<-, >=stealth, red] (s) -- (A2);
    \draw[<-, >=stealth, red, snake] (A2) -- (B3);
    \draw[<-, >=stealth, red] (B3) -- (B2);
    \draw[<-, >=stealth, red] (B2) -- (C2);
    \draw[<-, >=stealth, red, snake] (C2) -- (C3);
    \draw[<-, >=stealth, red] (C3) -- (D2);
    \draw[<-, >=stealth, red] (D2) -- (t);

    \draw (s) -- (A3);
    \draw (A2) -- (B2);
    \draw[snake] (A2) -- (B4);
    \draw (A3) -- (B3);
    \draw (B4) -- (C4);
    \draw (C4) -- (D4);
    \draw (D3) -- (D4);
    \draw[snake] (D3) -- (t);
    \draw (A0) -- (B1);

\end{tikzpicture}
    \caption{Shown above is a graph $H_F$ with its partial labels given a fixed $C$. The shaded vertices ($s$ and $t$) are already given a label. The red and blue edges represent $E_2$ which is the reversal of a flow graph $F$ of size 2. The black edges represent undirected unit-capacitated edges, each of which viewed as two directed edges of weight 1. The \emph{wavy} edges represent the fixing of $C$. Step 1 \emph{uniquely} labels vertices in $V(C)$ according to $C$. Such unique label is indicated by the label within each node itself. }
    \label{fig:step1}
\end{figure}
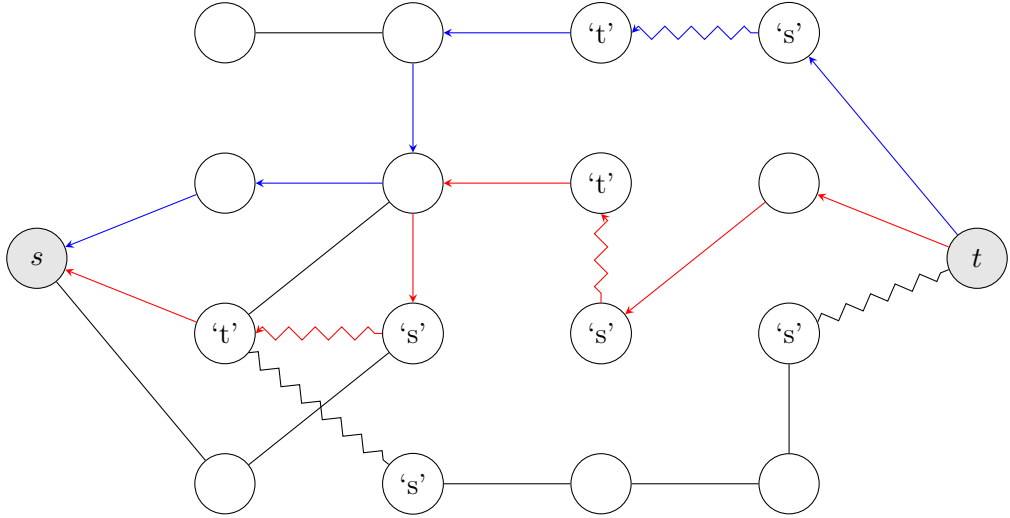

\paragraph{Step 2: Counting the number of labels for $V(E_2) \setminus V(C)$ given a label of $V(C)$.}  
Fix a labeling of $V(C)$ from Step 1. In this step, we will label $V(E_2)\setminus V(C)$. Recall that $E_2$ is the reversal of $F$, which consists of $f - k$ edge-disjoint $s$-$t$ flows. Thus, the subgraph induced by $E_2$ consists of $f - k$ edge-disjoint $t$-$s$ flows. This means we can partition $E_2$ into $f - k$ disjoint directed paths from $t$ to $s$. Call them $P_1, \ldots, P_{f - k}$. For each $i$, let $\eta_i = |P_i \cap C|$ denote the number of edges in $P_i$ that appear in $C$. We then have: $$\sum_{i=1}^{f - k} \eta_i = \sum_{i=1}^{f - k} |P_i \cap C| \leq |C| \leq w(C) \leq k - k' - 1.$$

Now consider each $i \in [f - k]$. The edges in $P_i \cap C$ split $P_i$ into $\eta_i + 1$ contiguous segments. We claim that for a segment consisting of $l$ edges, we can label its vertices in at most $l \leq n$ ways. To see this, suppose that this segment is $\mathcal{S} = (v_0, \ldots, v_l)$. Then, $v_0$ must have already been labeled with `t' (either $v_0 = t$ or it is the tail of some $s$-$t$ edge in $C$). Similarly, $v_{l}$ must have already been labeled with `s'. As there can be no other $s$-$t$ edges in $\mathcal{S}$, the only valid labelings are monotonic — that is, to label $\{v_0, \ldots, v_j\}$ with `t' and $\{v_{j+1}, \ldots, v_d\}$ with `s' for some $0 \leq j \leq l-1$; giving $l$ possible labeling of $V(\mathcal{S})$. See \Cref{fig:step2} for an illustration.

Thus, the number of consistent labels for $V(E_2)$ is upper-bounded by:  
$$\prod_{i=1}^{f - k} n^{\eta_i + 1} \leq n^{f - k + \sum_{i=1}^{f - k} \eta_i} \leq n^{f - k' - 1} < n^f.$$
Going forward, we fix the labels of the vertices in $V(E_2) \setminus V(C)$ accordingly.

\begin{figure}[h]
    \centering
    \begin{tikzpicture}

    \node[draw, circle, minimum size = 8mm, fill=gray!20] (s) at (-2.5, -1) {$s$};    

    \node[draw, circle, minimum size = 8mm] (A1) at (0, 0) {`s'};
    \node[draw, circle, minimum size = 8mm] (B2) at (2.5, 0) {`s'}; 
    \node[draw, circle, minimum size = 8mm] (B1) at (2.5, 2) {`t'};
    \node[draw, circle, minimum size = 8mm, fill=gray!20] (C1) at (5, 2) {`t'};
    \node[draw, circle, minimum size = 8mm, fill=gray!20] (D1) at (7.5,2) {`s'};    
    \node[draw, circle, minimum size = 8mm, fill=gray!20] (t) at (10, -1) {$t$};

    \draw[<-, >=stealth, blue] (s) -- (A1);
    \draw[<-, >=stealth, blue] (A1) -- (B2);
    \draw[<-, >=stealth, blue] (B2) -- (B1);
    \draw[<-, >=stealth, blue] (B1) -- (C1);
    \draw[<-, >=stealth, blue, snake] (C1) -- (D1);
    \draw[<-, >=stealth, blue] (D1) -- (t);

\end{tikzpicture}
    \caption{Continuing from Step 1, the shaded vertices are already given a fixed label. Shown above one of the 4 possible consistent labeling to a segment of length $l = 4$ on the blue flow.  The remaining 3 possibilities are as follows: `s' $\leftarrow$ `s' $\leftarrow$ `s'; `s' $\leftarrow$ `t' $\leftarrow$ `t'; `t' $\leftarrow$ `t' $\leftarrow$ `t'.}
    \label{fig:step2}
\end{figure}

\paragraph{Step 3: Counting the number of labels for $V$ given a labeling of $V(E_2) \cup V(C)$.}
Fix a labeling of $V(E_2) \cup V(C)$ from Steps 1 and 2. In this final steps, we will label the remaining vertices. For any vertex $u \in V$, let $d(u)$ denote the length of the shortest path from $s$ to $u$ in $H$ (such path must exist due to \Cref{assumption:H_connected}). We will prove the following claim by induction on $d$: For any vertex $u \in V$ such that $d(u) = d$, given a fixed labeling $\phi$ of $V(E_2) \cup V(C)$, there is at most one consistent label for $u$ with respect to $C$. We may assume $u \notin V(E_2) \cup V(C)$, as those vertices have already been assigned labels from the previous steps.

The base case $d = 1$ means $(s, u) \in E(H)$. Since $u \notin V(E_2)$, the undirected unit-capacitated edge $(s, u)$ belongs to $H_F$, which then can be replaced with two directed edges $(s, u, 1)$ and $(u, s, 1)$. Furthermore, because $u \notin V(C)$, neither $(s, u, 1)$ nor $(u, s, 1)$ can be an $s$-$t$ edge. As a result, $s$ and $u$ must be assigned the same label. Since $s$ is labeled with `s', $u$ can only be labeled with `s'.

For the inductive step, suppose the claim holds for all vertices $v$ with $d(v) \leq d - 1$. Consider a vertex $u$ with $d(u) = d$. Let $u'$ be the predecessor of $u$ on the shortest path from $s$ to $u$ in $H$, meaning $d(u') = d - 1$, and $(u', u)$ is an edge in $H$. By the induction hypothesis, $u'$ has a unique label, denoted as $l \in \{s, t\}$. We now consider the possible orientations of the edge(s) between $u$ and $u'$ in $H_F$:

\begin{enumerate}
    \item Either $(u, u', 2)$ or $(u', u, 2)$ belongs to $H_F$. In this case, $u \in V(E_2)$, which contradicts the assumption that $u \notin V(E_2)$.
    \item Both $(u, u', 1)$ and $(u', u, 1)$ belong to $H_F$. Since $u \notin V(C)$, neither $(u, u', 1)$ nor $(u', u, 1)$ is an $s$-$t$ edge. Therefore, $u$ and $u'$ must have the same label. As $u'$ can only take the label $l$, $u$ must also be labeled $l$.
\end{enumerate}

This completes the induction proof. As a result, given a labeling of $V(E_2) \cup V(C)$ from Steps 1 and 2, there can be at most one possible labeling of the remaining vertices in $V$. See \Cref{fig:step3} for an illustration.

\begin{figure}[h]
    \centering
    \begin{tikzpicture}
    \node[draw, circle, minimum size = 8mm, fill=gray!20] (s) at (-2.5, -1) {$s$};

    \node[draw, circle, minimum size = 8mm] (A0) at (0, 2) {`t'};
    \node[draw, circle, minimum size = 8mm, fill=gray!20] (A1) at (0, 0) {`s'};
    \node[draw, circle, minimum size = 8mm, fill=gray!20] (A2) at (0, -2) {`t'};
    \node[draw, circle, minimum size = 8mm] (A3) at (0, -4) {`s'};

    \node[draw, circle, minimum size = 8mm, fill=gray!20] (B1) at (2.5, 2) {`t'};
    \node[draw, circle, minimum size = 8mm, fill=gray!20] (B2) at (2.5, 0) {`s'};
    \node[draw, circle, minimum size = 8mm, fill=gray!20] (B3) at (2.5, -2) {`s'};
    \node[draw, circle, minimum size = 8mm, fill=gray!20] (B4) at (2.5, -4) {`s'};

    \node[draw, circle, minimum size = 8mm, fill=gray!20, fill=gray!20] (C1) at (5, 2) {`t'};
    \node[draw, circle, minimum size = 8mm, fill=gray!20] (C2) at (5, 0) {`t'};
    \node[draw, circle, minimum size = 8mm, fill=gray!20] (C3) at (5, -2) {`s'};
    \node[draw, circle, minimum size = 8mm] (C4) at (5, -4) {`s'};

    \node[draw, circle, minimum size = 8mm, fill=gray!20] (D1) at (7.5, 2) {`s'};
    \node[draw, circle, minimum size = 8mm, fill=gray!20] (D2) at (7.5, 0) {`t'};
    \node[draw, circle, minimum size = 8mm, fill=gray!20] (D3) at (7.5, -2) {`s'};
    \node[draw, circle, minimum size = 8mm] (D4) at (7.5, -4) {`s'};

    \node[draw, circle, minimum size = 8mm, fill=gray!20] (t) at (10, -1) {$t$};

    \draw[<-, >=stealth, blue] (s) -- (A1);
    \draw[<-, >=stealth, blue] (A1) -- (B2);
    \draw[<-, >=stealth, blue] (B2) -- (B1);
    \draw[<-, >=stealth, blue] (B1) -- (C1);
    \draw[<-, >=stealth, blue, snake] (C1) -- (D1);
    \draw[<-, >=stealth, blue] (D1) -- (t);

    \draw[<-, >=stealth, red] (s) -- (A2);
    \draw[<-, >=stealth, red, snake] (A2) -- (B3);
    \draw[<-, >=stealth, red] (B3) -- (B2);
    \draw[<-, >=stealth, red] (B2) -- (C2);
    \draw[<-, >=stealth, red, snake] (C2) -- (C3);
    \draw[<-, >=stealth, red] (C3) -- (D2);
    \draw[<-, >=stealth, red] (D2) -- (t);

    \draw (s) -- (A3);
    \draw (A2) -- (B2);
    \draw[snake] (A2) -- (B4);
    \draw (A3) -- (B3);
    \draw (B4) -- (C4);
    \draw (C4) -- (D4);
    \draw (D3) -- (D4);
    \draw[snake] (D3) -- (t);
    \draw (A0) -- (B1);

\end{tikzpicture}
    \caption{Continuing from Step 2, the shaded vertices are already given a fixed label. Shown above is the only possible label to the remaining (unshaded) vertices}
    \label{fig:step3}
\end{figure}

Combining the three steps, for a fixed $C$, there are at most $n^f$ possible labels that are consistent with $C$. This implies that the number of $(S, T)$ pairs whose cut is precisely $C$ is at most $n^f$. As we have already argued that there are $n^{O(f)}$ possibilities for $C$ and $X$, the total number of witnesses to $H_F$ is at most $n^{O(f)}$. This concludes the proof of \Cref{claim:witnesses_ub}.
\end{proof}

\subsection{Putting All Together: Proof of \Cref{claim:main_claim_comm}}
\label{subsec:conclude}

Finally, we conclude by proving \Cref{claim:main_claim_comm}. 

\mainclaimcomm*

\begin{proof}
Let $\beta > 0$ be a suitably small hidden constant following the statement of \Cref{claim:dec_witnesses}. We propose the following protocol.

\begin{enumerate}
\item{Players pre-process $H$ so that every vertex $v \in H$ is reachable from $s$.}
\item{Players locally determine $F \preceq H$, consisting of exactly $f - k$ edge-disjoint $s$-$t$ paths.}
\item If $\W_F(H) = \emptyset$, output the current $H$ and terminate. Otherwise, proceed.
\item[4a.] Alice declares her private edge $e^{\text{Alice}}$, if any, such that $|\W_F(H \cup e^{\text{Alice}})| \leq \left(1-\frac{\beta k'}{n\sqrt{k}}\right) \cdot |\W_F(H)|$.
\item[4b.] Bob declares his private edge $e^{\text{Bob}}$, if any, such that $|\W_F(H \cup e^{\text{Bob}})| \leq \left(1-\frac{\beta k'}{n\sqrt{k}}\right) \cdot |\W_F(H)|$.
\item[5.] Both players update $H \leftarrow H \cup \{ e^{\text{Alice}}, e^{\text{Bob}}\}$ and go back to Step 3.
\end{enumerate}

If our algorithm ever terminates, its correctness follows from \Cref{invariant_res} that $\W_F(H) = \emptyset$ implies $\nu(H) \geq f-k'$. Thus, it suffice to justify the termination and communication costs.

Assume that at the start of an iteration at Line 3, we have $\W_F(H) \ne \emptyset$. Applying \Cref{claim:dec_witnesses}, it is guaranteed that either Alice must have found $e^{\text{Alice}}$ or Bob must have found $e^{\text{Bob}}$. Without loss of generality, suppose that Alice has found $e^{\text{Alice}}$. Then, upon the insertion of $e^{\text{Alice}}$ and $e^{\text{Bob}}$, the number of remaining residual witnesses is reduced to:
\begin{align*}
    |\W_F(H \cup \{ e^{\text{Alice}}, e^{\text{Bob}}\})| & \leq |\W_F(H \cup \{ e^{\text{Alice}} \} | \tag{adding $e^{\text{Bob}}$ only reduces witnesses} \\
    & \leq \left(1-\frac{\beta k'}{n\sqrt{k}}\right) \cdot |\W_F(H)| \tag{Alice finds $e^{\text{Alice}}$ satisfying Line 4a}
\end{align*}
\Cref{claim:witnesses_ub} asserts that the number of starting residual witnesses is bounded by $n^{O(f)}$. Therefore, we only need 
$$\left(\frac{n\sqrt{k}}{\beta k'} \right) \cdot \log{\left(n^{O(f)}\right)} = \tO\left(nf \cdot \frac{\sqrt{k}}{k'}\right)$$ iterations before all residual witnesses are killed. Then, in the next iteration, the protocol will surely terminate.

Finally, in each iteration, the players only communicate at most two edges which costs $O(\log n)$ bits. Thus, the total communication cost of our protocol is, up to a polylogarithmic factor, bounded by the number of iterations which is $\tO\left(nf \cdot \frac{\sqrt{k}}{k'}\right)$. This concludes our proof. 
\end{proof}

\end{document}